\documentclass{CSML}

\def\dOi{12(3:13)2016}
\lmcsheading%
{\dOi}
{1--31}
{}
{}
{Apr.~\phantom05, 2016}
{Sep.~28, 2016}
{}

\ACMCCS{[{\bf Theory of computation}]: Logic}

\usepackage{amssymb}
\usepackage{amsmath}
\usepackage{mathrsfs}

\usepackage{color}

\usepackage{microtype}
\usepackage{bussproofs}
\usepackage{upgreek}

\usepackage{hyperref}
\usepackage{graphicx}

\begin{document}
\title[Curry-Howard Correspondence for Dummett's Logic]{On Natural Deduction for Herbrand Constructive Logics I: Curry-Howard Correspondence for Dummett's Logic LC}






\author[F.~Aschieri]{Federico Aschieri}
\address{Institut f\"ur Diskrete Mathematik und Geometrie\\Technische Universit\"at Wien\\ Wiedner Hauptstra\ss e 8-10/104, 1040, Vienna, Austria}
\thanks{This work was funded by the Austrian Science Fund FWF Lise Meitner grant M 1930--N35}








\newcommand*{\Cdot}[1][1.5]{%
  \mathpalette{\CdotAux{#1}} .%
}
\newdimen\CdotAxis
\newcommand*{\CdotAux}[3]{%
  {%
    \settoheight\CdotAxis{$#2\vcenter{}$}%
    \sbox0{%
      \raisebox\CdotAxis{%
        \scalebox{#1}{%
          \raisebox{-1.1pt}{%
            $\mathsurround=0pt #2#3$%
          }%
        }%
      }%
    }%
    \dp0=0pt %
    \sbox2{$#2\bullet$}%
    \ifdim\ht2<\ht0 %
      \ht0=\ht2 %
    \fi
    \sbox2{$\mathsurround=0pt #2#3$}%
    \hbox to \wd2{\hss\usebox{0}\hss}%
  }%
}

\newcommand{\comment}[1]{}

\newcommand{\AD}                       { {\mathsf{D^{\star}}} }
\newcommand{\redn}              {\succ}
\newcommand{\witn}[2]               {{\mathtt{W}^{\exists {#2} {#1}}}}
\newcommand{\red}           {\Vdash}

\newcommand{\Hypn}[2]                   {{\mathtt{H}^{\forall {#2} \mathsf{#1}}}}

\newcommand{\D}                       { {\mathsf{D}} }
\newcommand{\HA}                       { {\mathsf{HA}} }
\newcommand{\IL}                      {{\mathsf{IL}}}
\newcommand{\LC}     {{\mathsf{LC}}}
\newcommand{\ALC}     {{\mathsf{LC}^{\star}}}
\newcommand{\LCS}     {{\mathsf{LC}_{2}}}
\newcommand{\ALCS}     {{\mathsf{LC}_{2}^{\star}}}
\newcommand{\LJ}{\mathsf{LJ}}
\newcommand{\Language}                 {\mathcal{L}}
\newcommand{\E}[3]                   {{ #2 \parallel_{#1} #3}}
\newcommand{\Ez}[3]                   {{ #2\, |_{#1}\, #3}}
\newcommand{\proj}                     { {\mathsf{p}} }
\newcommand{\inj}                   {{{\upiota}}}
\newcommand{\emp}[1]    {{\mathsf{P}_{#1}}}
\newcommand{\hyp}                  {{\mathtt{H}_{\emp{}}}}
\newcommand{\Hyp}[3]                   {{\mathtt{H}_{#1}^{\forall {#3} {#2}}}}
\newcommand{\Hypz}[2]               {{\mathtt{H}_{#1}^{{#2}}}}
\newcommand{\wit}                 {{\mathsf{wit}_{\emp{}}}}
\newcommand{\Wit}[3]               {{\mathtt{W}_{#1}^{\exists {#3} {#2} }}}
\newcommand{\Witz}[2]               {{\mathtt{W}_{#1}^{{#2}}}}
\newcommand{\real}              {\, \Vdash\, }
\newcommand{\econt}[1]           {\mathcal{EM}[#1]}
\newcommand{\cruno}           {{\textbf{(CR1)}}}
\newcommand{\crdue}           {{\textbf{(CR2)}}}
\newcommand{\crtre}           {{\textbf{(CR3)}}}
\newcommand{\crquattro}           {{\textbf{(CR4)}}}
\newcommand{\crcinque}           {{\textbf{(CR5)}}}
\newcommand{\num}[1]{\overline{#1}}
\newcommand{\nSystemT}{\mathsf{T}^{\star}}
\newcommand{\nlambda}{\mathsf{\Lambda}^{\star}}
\newcommand{\SystemT}{\mathcal{T}}
\newcommand{\SystemTG}                  {\mathsf{T}}
\newcommand{\ifn}{{\mathsf{if}}}
\newcommand{\ifthen} [3]{ {\mathsf{if}\ {#1}\ \mathsf{then}\ {#2}\ \mathsf{else}\ {#3} } }
\newcommand{\Nat}                      { {\tt N} }
\newcommand{\Bool}                     { {\tt Bool} }
\newcommand{\sn}{\mathsf{HN}}
\newcommand{\snn}{\mathsf{HN}^{\star}}
\newcommand{\rec}                          {{\mathsf{R}}}
\newcommand{\True}                     { {\tt{True}} }
\newcommand{\False}                    { {\tt{False}} }
\newcommand{\suc}{\mathsf{S}}
\newcommand{\itr}{\mathsf{It}}
\newcommand{\eq}             {{\mathsf{eq}}}
\newcommand{\add}             {{\mathsf{add}}}
\newcommand{\mult}             {{\mathsf{mult}}}
\newcommand{\con}[1]          {{\mathsf{#1}}}
\newcommand{\seq}[1]       {{\vec{#1}}}

\newcommand{\pair}[2]{\langle #1,#2\rangle}
\newcommand{\prj}[2]{{[#2]}\pi_{#1}}
\newcommand{\const}{\mathsf{c}}
\newcommand{\trans}[1] {{#1}^{*}}
\newcommand{\nf}{\mathsf{NF}}
\newcommand{\hnf}{\mathsf{HNF}}
\newcommand{\nfn}{\mathsf{NF}^{*}}
\newcommand{\postnf}{\mathsf{NF}}
\newcommand{\gn}{\mathsf{GN}}
\newcommand{\prel}{\,\mathscr{P}\,}
\newcommand{\seqt}[3]{#1_{1}^{#3}\ldots #1_{#2}^{#3}}
\newcommand{\subst} [1]         { {\overline{#1}} }
\newcommand{\succrom}{\crom\hspace{-0.2ex}\mathbf{|}\hspace{-0.6ex}\mathbf{|}}
\newcommand{\usn}{\mathsf{UN}}
\newcommand{\Esucc}[3]{#2\, \succrom_{#1}\, #3}
\newcommand{\Ecrom}[3]{#2 \parallel_{#1} #3}
\newcommand{\crom}{\mathbf{|}\hspace{-0.6ex}\mathbf{|}\hspace{-0.6ex}\mathbf{|}\hspace{-0.6ex}\mathbf{|}\hspace{-0.6ex}\mathbf{|}\hspace{-0.6ex}\mathbf{|}\hspace{-0.6ex}\mathbf{|}\hspace{-0.6ex}\mathbf{|}}
\newcommand{\infi}[1]    {{\mathscr{#1}}}
\newcommand{\abort}    {{\mathcal{A}}}
\newcommand{\cand}[1] {\smash{\overset{\centerdot}{\mathcal{#1}}}}

\newcommand{\DL}   {{\mathsf{IL}+\mathsf{D}}}
\newcommand{\dott}  { {{\,\Cdot[2.4]\,}}}
\newcommand{\efq}[2]{{\mathsf{efq}_{#1}(#2)}}
\newcommand{\ILS}{{\mathsf{IL}^{2}}}
\newcommand{\sml}[1] {{\scriptscriptstyle #1}}


\newcommand{\linea}{\leavevmode\hrule\mbox{}}
\newcommand{\dlinea}{\leavevmode\hrule\vspace{1pt}\hrule\mbox{}}

\newcommand{\stefano}         {\textcolor{red}{\blacksquare}}
\newcommand{\evaluates}       { {\;\equiv\;} }


\newcommand{\str}[2]{{#1}_{#2}}
\newcommand{\pvar}[1]{\mathcal{#1}}






\begin{abstract}
Dummett's logic $\LC$ is intuitionistic logic extended with Dummett's axiom: for every two statements the first implies the second or the second implies the first. 
We present a natural deduction and a Curry-Howard correspondence for first-order and second-order Dummett's logic. We add to the lambda calculus an operator which represents, from the viewpoint of programming, a mechanism for representing parallel computations and communication between them, and from the viewpoint of logic, Dummett's axiom.  We prove that our typed calculus is normalizing and show that proof terms for existentially quantified formulas reduce to a list of individual terms forming an Herbrand disjunction. 
\end{abstract}
\maketitle



\section{Introduction} 
\label{section-introduction} 
{


We call \emph{Herbrand constructive} any intermediate logic -- a logic stronger than intuitionistic but weaker than classical -- which enjoys a strong form of Herbrand's theorem: for \emph{every} provable formula $\exists \alpha\, A$, the logic proves as well an Herbrand disjunction
 $$A[m_{1}/\alpha]\lor \ldots \lor A[m_{k}/\alpha]$$
Of course intuitionistic logic is trivially Herbrand constructive, but classical logic is not: $A$ is arbitrary! In between, there are several interesting logics which do have the property. Yet for Herbrand constructive logics there are no known natural deduction formulations with associated Curry-Howard correspondences, except in trivial cases.  We launch here a new series of papers to fill this void. 

We begin with Dummett's first-order and second-order logic $\mathsf{LC}$: intuitionistic logic extended with the so-called Dummett linearity axiom
$$(A\rightarrow B)\lor (B\rightarrow A)$$
$\LC$ was introduced by Dummett \cite{Dummett} as an example, in the propositional case, of a many-valued logic with a countable set of truth values. Its propositional fragment is also called G\"odel-Dummett logic, because it is based on the truth definition given in G\"odel's seminal paper on many-valued logics \cite{Godel}. In this case, the logic can be formalized by Corsi's sequent calculus \cite{Corsi} or by the more elegant  hypersequent calculus devised by Avron \cite{Avron}, \cite{Agata}. Surprisingly, Avron's hypersequent calculus does not work for first-order $\LC$: only recently Tiu \cite{Tiu} provided a more involved version of it, which indeed corresponds to $\LC$ at the first-order.

\subsection{Hyper Natural Deduction?}

In all this story, natural deduction is the great absent. Since it is one of the most celebrated logical deduction systems, the question is: how is that possible?

 The first issue is that $\LC$ is evidently a non-constructive system: for example, it proves the excluded middle for all negated formulas: $\lnot A\lor \lnot\lnot A$;  and Dummett's axiom poses even more problems. As it is well known, natural deduction was put aside by its own inventor, Gentzen, precisely for the reason that he was not able to prove a meaningful normalization theorem for classical natural deduction, whilst he \emph{was} for the intuitionistic case \cite{vonPlato}. It indeed took a surprisingly long time to discover suitable reduction rules for classical natural deduction systems with all connectives  \cite{deGroote}, \cite{AschieriZH}. Even this accomplishment, however, is still not enough:  although Dummett's axiom is classically provable, the known classical natural deduction systems fail to provide a refined computational interpretation of $\LC$. The trouble is that $\LC$ proofs are not closed under classical reductions, because during the normalization process instances of Dummett's axiom are replaced by reductio ad absurdum in $\lambda\mu$-calculus \cite{Parigot} and in \cite{DanosKrivine}, and by excluded middle in \cite{AschieriZH}.
  
 The second issue is that existential quantifiers are witnessed by multiple terms and so a parallel computational mechanism is desirable. No Curry-Howard correspondence offered a suitable one until very recently \cite{AschieriZH}. 
 
 Sequent calculus solves these issues by means of \emph{structural rules}. Classical logic is rendered by allowing more formulas on the righthand side of a sequent; Dummett's $\LC$ is rendered by allowing sequences of sequents and a communication mechanism between them. On the contrary, natural deduction usually solves the same issues by means of new \emph{reduction rules}. When one wants to add some new axiom to intuitionistic natural deduction, it is enough to add it straight away or as a rule, and all the ingenuity of the construction lies in the proof transformations associated to the axiom. 
 
Inspired by hypersequents, Baaz, Ciabattoni and Ferm\"uller \cite{HyperAgata} did not follow the latter path and changed instead the very structure of natural deduction into an hyper version. The resulting  logical calculus is  an \emph{hyper natural deduction} corresponding to G\"odel-Dummett first-order logic (which is not to be confused with first-order $\LC$ and can be axiomatized by adding to $\LC$ the axiom scheme $\forall \alpha\, (A\lor B)\rightarrow \forall \alpha\, A\,\lor\, B$, where $\alpha$ does not occur in $B$). The Normal Form Theorem, however, is only obtained by translation into the hypersequent calculus, followed by cut-elimination and backward translation: no reduction rules for hyper deductions were provided.   This last task was carried out by  Beckmann and Preining \cite{Beckmann}, who formulated a \emph{propositional} hyper natural deduction with a proof normalization procedure. Unfortunately, the structural rules are so complicated that the adjective ``natural'' does not fit any more. Another attempt along the ``hyper line'' has  been made by Hirai \cite{Hirai}, with the addition of an associated lambda calculus. One cannot speak of a Curry-Howard correspondence, however, because Subject reduction does not hold: there is no match between computational steps and proof reductions.  
 
 \subsection{Natural Deduction Again}
 Although hyper natural deduction is a legitimate proof system in its own right, the ``hyper approach'' is not the one we follow. 
 For two reasons. 
 
 The first reason is that we will show that natural deduction works perfectly as it is. There is no need to change its structure and, to render Dummett's axiom, it sufficient to add the inference rule
\begin{prooftree}
\AxiomC{$[A\rightarrow B]$}
\noLine
\UnaryInfC{$\vdots$}
\noLine
\UnaryInfC{$C$}
\AxiomC{$[B\rightarrow A]$}
\noLine
\UnaryInfC{$\vdots$}
\noLine
\UnaryInfC{$C$}
\RightLabel{$\D$}
\BinaryInfC{$ C$}
\end{prooftree}
which allows to conclude unconditionally $C$ from two different deductions of $C$: one from the hypothesis $A\rightarrow B$ and one from the  hypothesis $B\rightarrow A$.
We shall define simple reduction rules for proofs ending with this inference and we shall show that they are all we need to extract witnesses for existentially quantified formulas.

The second reason is that natural deduction should stay \emph{natural}. This is the very motivation that led to its discovery. Indeed, Gentzen starts his celebrated work \cite{Gentzen} on natural deduction and sequent calculus complaining that the proof systems known at the time were far removed from the actual mathematical reasoning. And his main goal was to set up a formalism with the aim of ``\emph{reproducing as precisely as possible the real logical reasoning in mathematical proofs}''. To avoid betraying natural deduction's philosophical motivations,  there is no alternative but to add an inference rule that naturally mirrors the kind of reasoning corresponding to Dummett's axiom, which is our approach. 

\subsection{Realizability}

One of the most attractive features of intuitionistic natural deduction is that, in a very precise sense, it does not need a truth-based semantics. Logical inferences are divided into two groups: introduction rules and elimination rules. And as Gentzen \cite{Gentzen} himself famously suggested, introduction rules \emph{define}, so to speak,  the meaning of the logical constants they introduce; elimination rules, on the other hand, are nothing but \emph{consequences} of these definitions. In other words, introduction rules are self-justifyng, because they fix themselves the meaning of their conclusions, whereas elimination rules are sound in virtue of the meaning fixed by the introductions. For example, the rule 
\begin{prooftree}
 \AxiomC{$[A]$}
\noLine
\UnaryInfC{$\vdots$}
\noLine
\UnaryInfC{$B$}
\UnaryInfC{$A\rightarrow B$}
\end{prooftree}
says that the grounds for asserting $A\rightarrow B$ consist in a proof of $B$ from the hypothesis $A$;  therefore,  the elimination
$$\begin{array}{c}    A\rightarrow B\ \ \  \ \ A \\ \hline
 B
\end{array} $$
is automatically justified: if we have a proof of $A$ we can plug it into the proof of $B$ from $A$, whose existence is warranted by the meaning of  $A\rightarrow B$, and obtain a proof of $B$. The reverse approach works as well: we may consider elimination rules as meaning constitutive and treat introduction rules as consequences of the meaning fixed by eliminations. In other words, meaning is determined by \emph{how we use} a statement, by what we can directly \emph{obtain} from the statement; we shall adopt this \emph{pragmatist} standpoint, elaborated by Dummett himself \cite{DummettL}.

This idea of internal justification, as it is,  cannot be generalized straight away for extensions of intuitionistic logic: new inferences tend to break the harmony between introductions and eliminations. It is at this point that Brouwer's view of logic comes into play.  According to Brouwer \cite{Brouwer}, the string of ``logical'' steps appearing in a mathematical proof is in reality a sequence of mathematical constructions. What we perceive as inference rules are instead transformations of constructions for the premises into constructions for the conclusion. This insight finds a precise formalization by means of the Curry-Howard isomorphism: a proof is indeed isomorphic to an effective construction, in fact, \emph{it is}, in and of itself,  {a construction}. 

Since proofs are constructions, the role of semantics is just \emph{explaining what these constructions do}.  Hence, a proof-theoretic semantics of an intermediate logic is in principle always possible and is made of two ingredients: a formalization of \emph{proofs as programs} and a semantical description of what these programs achieve with their calculations. The first is obtained through the decoration of deduction trees with lambda terms, the second is the task of \emph{realizability}.

Realizability was introduced by Kleene \cite{Kleene} to computationally interpret intuitionistic first-order Arithmetic, but it is Kreisel's \cite{Kreiselm} later version with typed terms which embodies the modern perspective on the subject. Though it was initially conceived just for intuitionistic theories, realizability can be extended to intuitionistic Arithmetic with Markov's principle \cite{AZMarkov}, to intuitionistic Arithmetic with the simplest excluded middle $\mathsf{EM}_{1}$ \cite{ABB} and even all the way up to the strongest classical theories \cite{AschieriCSL, AschieriHAS, Krivine}. Realizability  replaces the notion of truth with the notion of \emph{constructive evidence}. A formula holds if it is \emph{realized} by some typed program, providing some constructive information about the formula.

 In the following, we shall build a realizability interpretation for Dummett's $\LC$, inspired by Krivine's realizability \cite{Krivine, DanosKrivine}. By construction, every realizer always terminates its computations and, in particular, whenever it realizes an existentially quantified formula $\exists\alpha\, A$, it reduces to a term of the shape
$$(m_{0}, v_{0})\parallel_{a_{1}} (m_{1}, v_{1})\parallel_{a_{2}}\ldots \parallel_{a_{k}} (m_{k}, v_{k})$$
with the property that 
$$\LC \vdash A[m_{1}/\alpha]\lor \dots \lor A[m_{k}/\alpha]$$
The circle is closed by a soundness theorem, the Adequacy Theorem: every formula provable in $\LC$ is realized by a closed program, which immediately implies the Normalization Theorem -- every proof reduces to a normal form -- and that $\LC$ is Herbrand constructive. Therefore, to extract an Herbrand disjunction it suffices to reduce any proof of any existentially quantified formula to a normal form, according to a very simple set of reduction rules. 

\subsection{Reduction Rules}

To find simple and terminating reduction rules for a natural deduction system is always tricky, but once the job is done, the reductions often look so natural that they appear inevitable. 
It is the effort of removing obstacles toward a good normal form what inevitably leads to these  reductions, as the flow of a river leads to the sea. 
In the case of $\LC$, the main obstacles toward witness extraction for a formula $\exists \alpha\, A$ are  configurations in which one of the hypotheses introduced by the Dummett inference blocks the reduction. For example, let us consider this proof shape:
\vspace{-2ex}
\begin{prooftree}
\AxiomC{$[A\rightarrow B]$}
\AxiomC{$$}
\noLine
\UnaryInfC{$\vdots$}
\noLine
\UnaryInfC{$A$}
\BinaryInfC{$B$}
\doubleLine\RightLabel{$\mathsf{EL}$}
\UnaryInfC{$\exists \alpha\, C$}
\noLine
\AxiomC{$[B\rightarrow A]$}
\UnaryInfC{$\mathcal{D}$}
\noLine
\UnaryInfC{$\exists \alpha\, C$}
\RightLabel{$\D$}
\BinaryInfC{$\exists\alpha\, C$}
\end{prooftree}
where $\exists \alpha\, C$ has been obtained from $B$ by a series of elimination rules. It is clear that no witness can be retrieved in the left branch of the proof above, because there is just a proof of $A$ and, magically, a ``void'' proof of $B$ obtained by modus ponens from $A$ and the arbitrary hypothesis $A\rightarrow B$. But can't we  just send the proof of $A$ to the right branch of the Dummett rule and obtain a direct proof of $\exists \alpha\, C$, like this? 
\vspace{-3ex}
\begin{prooftree}
\AxiomC{$$}
\noLine
\UnaryInfC{$\vdots$}
\noLine
\UnaryInfC{$A$}
\UnaryInfC{$B\rightarrow A$}\
\noLine
\UnaryInfC{$\mathcal{D}$}
\noLine
\UnaryInfC{$\exists \alpha\, C$}
\end{prooftree}
No! In fact, the proof of $A$ too might depend on the hypothesis  $A\rightarrow B$, so that the original proof could be
\begin{prooftree}
\AxiomC{$[A\rightarrow B]$}
\AxiomC{$[A\rightarrow B]$}
\noLine
\UnaryInfC{$\vdots$}
\noLine
\UnaryInfC{$A$}
\insertBetweenHyps{\hskip -3pt}
\BinaryInfC{$B$}
\doubleLine\RightLabel{$\mathsf{EL}$}
\UnaryInfC{$\exists \alpha\, C$}
\noLine
\AxiomC{$[B\rightarrow A]$}
\UnaryInfC{$\mathcal{D}$}
\noLine
\UnaryInfC{$\exists \alpha\, C$}
\RightLabel{$\D$}
\BinaryInfC{$\exists\alpha\, C$}
\end{prooftree}
and thus the previous transformation is unsound. But the idea of sending the proof of $A$ to the right branch can work if the right branch is in turn moved on the left like this
\begin{prooftree}
\AxiomC{$[A\rightarrow B]$}
\noLine
\UnaryInfC{$\vdots$}
\noLine
\UnaryInfC{$A$}
\UnaryInfC{$B\rightarrow A$}\
\noLine
\UnaryInfC{$\mathcal{D}$}
\noLine
\UnaryInfC{$\exists \alpha\, C$}
\noLine
\AxiomC{$[B\rightarrow A]$}
\UnaryInfC{$\mathcal{D}$}
\noLine
\UnaryInfC{$\exists \alpha\, C$}
\RightLabel{$\D$}
\BinaryInfC{$\exists\alpha\, C$}
\end{prooftree}
The reductions that we shall give  generalize this transformation in order to work in every situation. 
\subsection{Curry-Howard Correspondence}

It is more convenient to express proof reductions in terms of program reductions, because for that purpose the lambda notation is superior to the proof tree notation. For this reason, we shall define a lambda calculus isomorphic to natural deduction for $\LC$ and  then define an \emph{head reduction strategy} for lambda terms, inspired by Krivine's strategy \cite{Krivine}. The termination of head reduction will just be a consequence of soundness of $\LC$ with respect to realizability, while the perfect match between program reductions and proof reductions will as usual be consequence of the Subject Reduction Theorem. The decoration of intuitionistic inferences with programs is standard and Dummett's rule will be decorated in the following way
\begin{prooftree} 
\AxiomC{$[a^{\scriptscriptstyle A\rightarrow B}: A\rightarrow B]$}
\noLine
\UnaryInfC{$\vdots$}
\noLine
\UnaryInfC{$u: C$}
\AxiomC{$[a^{\scriptscriptstyle B\rightarrow A}: B\rightarrow A]$}
\noLine
\UnaryInfC{$\vdots$}
\noLine
\UnaryInfC{$v: C$}
\RightLabel{$\D$}
\BinaryInfC{$u \parallel_{a} v: C$}
\end{prooftree}
The parallel operator $\parallel_{a}$ is inspired by the exception operator studied in \cite{AschieriZH} and keeps using the variable $a$ for communication purposes. The variable $a$ has the task of sending terms from $u$ to $v$ and viceversa, as well as allowing $u$ to call the process $v$ whenever it needs it and viceversa.


}
\subsection{Plan of the Paper}

 In Section \S \ref{section-system} we introduce a Curry-Howard interpretation of intuitionistic first-order natural deduction extended with the Dummett rule $\D$. We first describe the calculus together with its computational rules and then discuss its proof theoretical interpretation. \\ In Section \S \ref{section-systemNEM} we prove the Normalization Theorem and the soundness of realizability with respect to $\LC$. \\In Section \S \ref{sec:nfhd}, we prove that $\LC$ is Herbrand constructive and in particular
 that from any closed term having as type an existentially quantified formula, one can  extract a corresponding Herbrand disjunction.\\
 In Section \S \ref{section-secondorder} we extend the previous results to the second-order $\LCS$, achieving its first computational interpretation, for there is no known cut-elimination procedure for second-order hypersequent calculus.

\section{The System $\LC$}
\label{section-system}
In this section we describe a standard natural deduction system for intuitionistic first-order logic, with a term assignment based on the Curry-Howard correspondence (e.g. see \cite{Sorensen}), and add on top of it an operator which formalizes Dummett's axiom. First, we shall describe the lambda terms and their computational behavior, proving as main result the Subject Reduction Theorem, stating that the reduction rules preserve the type. Then, we shall analyze the logical meaning of the reductions and present them as pure proof transformations.


We start with the standard first-order language of formulas.

\begin{defi}[Language of $\LC$]\label{definition-languagear}
The language $\Language$ of $\LC$ is defined as follows.
\begin{enumerate}

\item
The \textbf{terms} of $\Language$ are inductively defined as either variables $\alpha, \beta,\ldots$ or constants  $\con{c}$ or expressions of the form $\con{f}(m_{1}, \ldots, m_{n})$, with $\con{f}$ a function constant of arity $n$ and $m_{1}, \ldots, m_{n}\in\Language$.  

\item
There is a countable set of \textbf{predicate symbols}. The \emph{atomic formulas} of $\Language$ are all the expressions of the form $\mathcal{P}(m_{1}, \ldots, m_{n})$ such that  $\mathcal{P}$ is a predicate symbol of arity $n$ and $m_{1}, \ldots, m_{n}$ are terms of $\Language$. We assume to have a $0$-ary predicate symbol $\bot$ which represents falsity.

\item
The \textbf{formulas} of $\Language$ are built from atomic formulas of $\Language$ by the logical constants  $\lor,\land,\rightarrow, \forall,\exists$, with quantifiers ranging over  variables $\alpha, \beta, \ldots$: if $A, B$ are formulas, then $A\land B$, $A\lor B$, $A\rightarrow B$, $\forall \alpha\, A$, $\exists \alpha\, B$ are formulas. The  logical negation $\lnot A$ can be introduced, as usual, as a shorthand  for the formula $A\rightarrow\bot$.

\end{enumerate}

\end{defi}

\bigskip

 In Figure \ref{fig:system} we define a type assignment for lambda terms, called \textbf{proof terms}, which is isomorphic to natural deduction for intuitionistic logic extended with Dummett's axiom. 
 \begin{figure*}[!htb]
 
\footnotesize{
\dlinea

\begin{description}

\comment{We also assume that the term formation rules are applied in such a way that in each term $t$, if $t$ contains $\Wit{a}{P}{{\alpha}}$ or $\Hyp{a}{P}{{\alpha}}$ and $t$ contains $\Wit{a}{Q}{{\alpha}}$ or $\Wit{a}{Q}{{\alpha}}$,  then $\mathsf{P}=\mathsf{Q}$.}


\comment{or $\wit \beta$ (for some individual variable $\beta$) and $A_1,\ldots, A_n$ formulas of $\Language$.}
\item[Axioms] 
$\begin{array}{c}    x^A: A
\end{array}\ \ \ \ $
\\

\item[Conjunction] 
$\begin{array}{c}  u:  A\ \ \ \  t: B\\ \hline  \langle
u,t\rangle:
A\wedge B
\end{array}\ \ \ \ $
$\begin{array}{c}  u: A\wedge B\\ \hline u\,\pi_0: A
\end{array}\ \ \ \ $
$\begin{array}{c}   u: A\wedge B\\ \hline  u\,\pi_1 : B
\end{array}$
\vskip 0.15in
\item[Implication] 
$\begin{array}{c}   t: A\rightarrow B\ \ \  u:A \\ \hline
 t u:B
\end{array}\ \ \ \ $
 \AxiomC{$[x^{A}: A]$}
\noLine
\UnaryInfC{$\vdots$}
\noLine
\UnaryInfC{$u: B$}
\UnaryInfC{$\lambda x^{A} u: A\rightarrow B$}
\DisplayProof
\vskip 0.15in
\item[Disjunction Introduction] 
$\begin{array}{c} u: A\\ \hline  \inj_{0}(u): A\vee B
\end{array}\ \ \ \ $
$\begin{array}{c}   u: B\\ \hline \inj_{1}(u): A\vee B
\end{array}$
\vskip 0.15in

\item[Disjunction Elimination] 
\AxiomC{$u: A\lor B$}
\AxiomC{$[x^{A}: A]$}
\noLine
\UnaryInfC{$\vdots$}
\noLine
\UnaryInfC{$w_{1}: C$}
\AxiomC{$[y^{B}: B]$}
\noLine
\UnaryInfC{$\vdots$}
\noLine
\UnaryInfC{$w_{2}: C$}
\TrinaryInfC{$u\, [x^{A}.w_{1}, y^{B}.w_{2}]: C$}
\DisplayProof
\vskip 0.15in

\item[Universal Quantification] 
$\begin{array}{c} u:\forall \alpha\, A\\ \hline   u m: A[m/\alpha]
\end{array}\ \ \ $
$\begin{array}{c}   u: A\\ \hline  \lambda \alpha\, u:
\forall \alpha\, A
\end{array}$\\
\vskip 0.1in
where $m$ is any term of  the language $\Language$ and $\alpha$ does not occur
free in the type $B$ of any free variable  $x^{B}$ of $u$.\\
\vskip 0.1in

\item[Existential Quantification] 
$\begin{array}{c}  u: A[m/\alpha]\\ \hline  (
m,u):
\exists
\alpha\, A
\end{array}\ \ \ $
\AxiomC{$u: \exists \alpha\, A$}
\AxiomC{$[x^{A}: A]$}
\noLine
\UnaryInfC{$\vdots$}
\noLine
\UnaryInfC{$t: C$}
\BinaryInfC{$u\, [(\alpha, x^{A}). t]: C$}
\DisplayProof
\\\vskip 0.1in
where $\alpha$ is not free in $C$ nor in the type $B$ of any free variable of $t$.\\
\vskip 0.15in

\item[Dummett's Axiom $\D$]
\AxiomC{$[a^{\scriptscriptstyle A\rightarrow B}: A\rightarrow B]$}
\noLine
\UnaryInfC{$\vdots$}
\noLine
\UnaryInfC{$u: C$}
\AxiomC{$[a^{\scriptscriptstyle B\rightarrow A}: B\rightarrow A]$}
\noLine
\UnaryInfC{$\vdots$}
\noLine
\UnaryInfC{$v: C$}
\RightLabel{$\D$}
\BinaryInfC{$u \parallel_{a} v: C$}
\DisplayProof
\vskip 0.15in
\item[Ex Falso Quodlibet] 
$\begin{array}{c}  \Gamma \vdash u: \bot \\ \hline \Gamma\vdash  \efq{P}{u}:
P
\end{array}$\\
with $P$ atomic.

\end{description}
}

\dlinea
\caption{Term Assignment Rules for $\LC$}\label{fig:system}
\end{figure*}

 \comment{The EM1-rule $\D{a}{u}{v} : C$ takes the form of an elimination rule for $\vee$ applied to the axiom $\forall \alpha. \mathsf{P} \vee \exists \alpha. \lnot \mathsf{P}$: however, we do not explicitly write the axiom $\forall \alpha. \mathsf{P} \vee \exists \alpha. \lnot \mathsf{P}$, and we only write the proofs $u$, $v$ of $C$ from $\forall \alpha. \mathsf{P}$ and from $\exists \alpha. \lnot \mathsf{P}$. These two assumptions are both associated to $a$, in the way explained below.}

We assume that in the proof terms two distinct classes of variables appear. The first class of variables is made by the variables for the proof terms themselves: for every formula $A$, we have variables $x_{0}^{A}, x_{1}^{A}, x_{2}^{A}, \ldots$ of type $A$; these variables will be denoted as $x^{A}, y^{A}, z^{A}\ldots, a^{A}, b^{A}$ and whenever the type is not important simply as $x, y, z, \ldots, a, b$. For clarity, the variables introduced by the Dummett's inference rule will be denoted with letters $a, b, \ldots$, but they are not in any syntactic category apart. The second class of variables is made by the quantified variables of the formula language $\Language$ of $\LC$, denoted usually as $\alpha, \beta, \ldots$. 

The free and bound variables of a proof term are defined as usual and for  the new  term $\Ecrom{a}{u}{v}$, all of the free occurrences of $a$ in $u$ and $v$ are bound in $\Ecrom{a}{u}{v}$. In the following, we assume the standard renaming rules and alpha equivalences that are used to avoid capture of variables in the reduction rules that we shall give.

 Whenever $\Gamma= x_{1}: A_{1}, \ldots, x_{n}: A_{n}$ and the list $x_{1}, \ldots, x_{n}$ includes all the free variables of a proof term $t: A$, we shall write $\Gamma\vdash t: A$. From the logical point of view, the notation means that $t$ represents a natural deduction of $A$ from the hypotheses $A_{1}, \ldots, A_{n}$. We shall write $\LC\vdash t: A$ whenever $\vdash t: A$, and the notation means provability of $A$ in intuitionistic logic with Dummett's axiom.

We are now going to explain the basic reduction rules for the proof terms of $\LC$, which are given in Figure \ref{fig:red}. To understand them, we need the notions of parallel context and stack. If
we omit parentheses, any term $t$  can be written, not uniquely, in the form
 $$t = t_{1}\parallel_{a_{1}} t_{2}\parallel_{a_{2}}\ldots \parallel_{a_{n}} t_{n+1}$$
If we replace some $t_{i}$ with a ``hole'' $[]$ to be filled, the expression above becomes  a parallel context.  



\begin{defi}[Parallel Contexts]\label{defi-parallelc}
Omitting parentheses, a \textbf{parallel context} $\mathcal{C}[\ ]$ is an expression of the form 
$$u_{1}\parallel_{a_{1}} u_{2}\parallel_{a_{2}}\ldots u_{i} \parallel_{a_{i}} [] \parallel_{a_{i+1}}u_{i+1}\parallel_{a_{i+2}}\ldots \parallel_{a_{n}} u_{n}$$
where $[]$ is a placeholder and $u_{1}, u_{2}, \ldots,  u_{n}$ are proof terms. 
For any proof term $u$, $\mathcal{C}[u]$ denotes the replacement in $\mathcal{C}[\ ]$ of the placeholder $[]$ with $u$:
$$u_{1}\parallel_{a_{1}} u_{2}\parallel_{a_{2}}\ldots u_{i} \parallel_{a_{i}} u \parallel_{a_{i+1}}u_{i+1}\parallel_{a_{i+2}}\ldots \parallel_{a_{n}} u_{n}$$
\end{defi}

A stack represents, from the logical perspective, a series of elimination rules; from the lambda calculus perspective, a series of either operations to be performed or arguments to be given as input to some program. A stack is  also known as a \emph{continuation}, because it embodies a series of tasks that wait to be executed, and corresponds to Krivine's stacks \cite{Krivine}.
\begin{defi}[Stack]\label{definition-stack}
A \textbf{stack} is a sequence $$\sigma = \sigma_{1}\dott\sigma_{2}\dott\ldots\dott \sigma_{n} $$
such that for every $	1\leq i\leq n$, exactly one of the following holds:
\begin{itemize}
\item $\sigma_{i}=t$, with $t$ proof term.
\item $\sigma_{i}=m$, with $m\in\Language$.
\item $\sigma_{i}=\pi_{j}$, with $j\in\{0,1\}$.
\item $\sigma_{i}=[x.u, y.v]$, with $u, v$ proof terms of the same type.
\item $\sigma_{i}=[(\alpha, x).v]$, with $v$ proof term.
 \end{itemize} 
 If no confusion with other sequences of terms arises, $\sigma$ will often be written without intermediate dots, that is, as $\sigma_{1}\, \sigma_{2}\, \ldots\, \sigma_{n}$. The \emph{empty sequence} is denoted with $\epsilon$ and with $\xi, \xi', \ldots$ we will denote stacks of length $1$. If $t$ is a proof term, as usual in lambda calculus $t\, \sigma$ denotes the term $(((t\, \sigma_{1})\,\sigma_{2})\ldots \sigma_{n})$.
 \end{defi}

We find among the reductions in Figure \ref{fig:red} the ordinary reductions for the intuitionistic constructs together with Prawitz-style permutation rules ~\cite{Prawitz} for $\D$, as in \cite{AschieriZH}.
The reduction rules for $\D$ model the communication mechanism explained in Section \S\ref{section-introduction}. In the reduction
$$\mathcal{C}[a^{\scriptscriptstyle A\rightarrow B}\, u\,\sigma]\parallel_{a} v\ \mapsto\  \mathcal{C}[\,v[\lambda y^{\scriptscriptstyle B}\, u /a^{\scriptscriptstyle B\rightarrow A}]\,]\parallel_{a} v$$
we see that the term on the left is in some way stuck:
the variable $a^{\scriptscriptstyle A\rightarrow B}$ faces an argument $u$ of type $A$; of course, it has no idea how to use $u$ to produce a term of type $B$! On the contrary, the term $v$ knows very well how to use $u$ to produce something useful, because it contains the variable $a^{\scriptscriptstyle B\rightarrow A}$, which waits for a term of type $B\rightarrow A$. Thus, $a^{\scriptscriptstyle A\rightarrow B}$ sends the term  $\lambda y^{\scriptscriptstyle B}\, u$, with $y$ dummy, to $v$, yielding the term $\,v[\lambda y^{\scriptscriptstyle B}\, u /a^{\scriptscriptstyle B\rightarrow A}]$. This program is called to replace the useless $a^{\scriptscriptstyle A\rightarrow B}\, u\,\sigma$ and computation can go ahead. We require the context $\mathcal{C}[\ ]$ to be parallel, because in this way types are not needed to define the reductions for $\D$ and the calculus makes sense also in its untyped version and with Curry-style typing. We have chosen Church-typing only to make clearer the intended meaning of the operations:  had we omitted all the types from the terms, everything would have still worked just fine. In Theorem \ref{subjectred}, we shall prove that indeed our reduction rules for $\D$ are logically correct and preserve the type.
 
 
 

 \begin{figure*}[!htb]
\footnotesize{
\dlinea
\begin{description}

\item[Reduction Rules for Intuitionistic Logic]
\[(\lambda x\, u)t\mapsto u[t/x]\] 
\[(\lambda \alpha\, u)m\mapsto u[m/\alpha]\]
 \[ \pair{u_0}{u_1}\pi_{i}\mapsto u_i, \mbox{ for $i=0,1$}\]
\[\inj_{i}(u)[x_{1}.t_{1}, x_{2}.t_{2}]\mapsto t_{i}[u/x_{i}], \mbox{ for $i=0,1$} \]
\[(m, u)[(\alpha,x).v]\mapsto v[m/\alpha][u/x], \mbox{ for each term $m$ of $\Language$} \]




\item[Permutation Rules for $\D$]
\[(\Ecrom{a}{u}{v}) w \mapsto \Ecrom{a}{uw}{vw},\mbox{ if $a$ does not occur free in $w$} \]
\[(\Ecrom{a}{u}{v})\pi_{i}  \mapsto \Ecrom{a}{u\pi_{i}}{v\pi_{i}} \]
\[(\Ecrom{a}{u}{v})[x.w_{1}, y.w_{2}] \mapsto \Ecrom{a}{u[x.w_{1}, y.w_{2}]}{v[x.w_{1}, y.w_{2}]},\mbox{ if $a$ does not occur free in $w_{1},w_{2}$}
\]
\[(\Ecrom{a}{u}{v})[(\alpha, x).w] \mapsto \Ecrom{a}{u[(\alpha, x).w]}{v[(\alpha, x).w]}, \mbox{ if $a$ does not occur free in $w_{1},w_{2}$}
\]

\item[Reduction Rules for $\D$]
\[\mathcal{C}[a^{\scriptscriptstyle A\rightarrow B}\, u\,\sigma]\parallel_{a} v\ \mapsto\  \mathcal{C}[\,v[\lambda y^{\scriptscriptstyle B}\, u /a^{\scriptscriptstyle B\rightarrow A}]\,]\parallel_{a} v\]
\[v\parallel_{a}\mathcal{C}[a^{\scriptscriptstyle A\rightarrow B}\, u\,\sigma] \mapsto\  v\parallel_{a}\mathcal{C}[\,v[\lambda y^{\scriptscriptstyle B}\, u /a^{\scriptscriptstyle B\rightarrow A}]\,]\]
\[\mbox{for some parallel context $\mathcal{C}$, stack $\sigma$, variable $a$ free in $\mathcal{C}[a^{\scriptscriptstyle A\rightarrow B}\, u\,\sigma]$, dummy variable ${y}$ not occurring in $u$} \]

\comment{
\item[Reduction Rules for $\D$]
\[\left(u_{0}\parallel_{a_{1}}\ldots \parallel_{a_{i}} a\, t\,\sigma \parallel_{a_{i+1}}\ldots \parallel_{a_{n}} u_{n}\right)\, \parallel_{a}\, v\ \mapsto\  \left(u_{0}\parallel_{a_{1}}\ldots \parallel_{a_{i}} v[(\lambda y\, t)/a] \parallel_{a_{i+1}}\ldots \parallel_{a_{n}} u_{n}\right)\, \parallel_{a}\, v\]
\[v\, \parallel_{a}\, \left(u_{0}\parallel_{a_{1}}\ldots \parallel_{a_{i}} a\, t\,\sigma \parallel_{a_{i+1}}\ldots \parallel_{a_{n}} u_{n}\right)\ \mapsto\ v\, \parallel_{a}\, \left(u_{0}\parallel_{a_{1}}\ldots \parallel_{a_{i}} v[(\lambda y\, t)/a] \parallel_{a_{i+1}}\ldots \parallel_{a_{n}} u_{n}\right)\]
\[\mbox{ whenever $a\, t\, \sigma$ is the leftmost head redex of the whole term  and}\]
\[\mbox{${y}$ is a variable not occurring in $t$. 
}\]}
\end{description}}
\dlinea
\caption{Basic Reduction Rules for $\LC$}\label{fig:red}
\end{figure*}

Our goal now is to define a reduction strategy for typed terms of $\LC$: a recipe for selecting, in any given term, the subterm to which apply one of our basic reductions. As most typed lambda calculi are strongly normalizing and our reduction rules look fairly innocuous, one cannot help but conjecture that any reduction strategy eventually terminates; in other words, that reduction strategies are not necessary. We do conjecture that the fragment with $\forall, \rightarrow, \land, \lor$ is indeed strongly normalizing.
 Yet, already the proof of this weaker result appears excessively complex, to such an extent that arbitrary reduction strategies start to feel wrong, that is, to perform unnecessary computations. 
 
 We therefore leave strong normalization as an open problem and follow a more standard approach: Krivine's (weak) head reduction strategy. The difference is: in Krivine's calculus each process has a unique head; in our calculus each process has several heads, like the Hydra monster. This is due to  the presence of the parallel operator $\parallel_{a}$. Indeed, if we omit parenthesis, any term $t$  can be written, not uniquely, in the form
 $$t = t_{1}\parallel_{a_{1}} t_{2}\parallel_{a_{2}}\ldots \parallel_{a_{n}} t_{n+1}$$
The terms $t_{1}, \ldots, t_{n}$ are parallel processes; each one has its own head and may have an head redex. And as with the Hydra monster, if we contract some head $t_{i}$, more heads to contract might grow. We now formally define what are the parallel processes that appear in a term and what is the head redex of a term.  
 \comment{
 \subsubsection{A Counterexample to Strong Normalization}
Let us consider the following term:
$$u:= \lambda x^{\scriptscriptstyle A\lor B}\, x\, [x_{1}^{\scriptscriptstyle A}.w_{1}, x_{2}^{\scriptscriptstyle B}.w_{2}]\, x: (A\lor B)\rightarrow C$$
where $w_{1}$ and $w_{2}$ are of type $(A\lor B)\rightarrow C$
}

\begin{defi}[Parallel Processes, Head]\label{definition-head}
\mbox{}
\begin{itemize}

\item Removing the parentheses,  whenever a proof term $t$  can be written as
$$t = t_{1}\parallel_{a_{1}} t_{2}\parallel_{a_{2}}\ldots \parallel_{a_{n}} t_{n+1}$$
each term $t_{i}$, for $1\leq i\leq n+1$,  is said to be a \textbf{parallel process} of $t$ and is said to be an \textbf{elementary process} of $t$ in case it is not of the form $u\parallel_{a}v$. 

\item A \textbf{redex} is a term $u$ such that $u\mapsto v$ for some $v$ and basic reduction of Figure \ref{fig:red}.  

\item Let $\sigma$ be any stack. A \emph{redex} $h$ is said to be the \textbf{head redex} of a proof term $t$ in the following cases:
\begin{enumerate}
\item $t=(\lambda x\, u) v\,\sigma$ and $h=(\lambda x\, u) v$; 
\item $t=(\lambda \alpha\, u) m\,\sigma$ and $h=(\lambda \alpha\, u) m$; 
\item $t=\pair{u}{v} \pi_{i}\,\sigma$ and $h=\pair{u}{v} \pi_{i}$; 
\item $t=\inj_{i}(u)\, [x_{1}.t_{1}, x_{2}.t_{2}]\, \sigma$ and $h=\inj_{i}(u)\, [x_{1}.t_{1}, x_{2}.t_{2}]$; 
\item $t=(m, u)\, [(\alpha, x).v]\, \sigma$ and $h=(m, u)\, [(\alpha, x).v]$; 
\item $t=((u\parallel_{a} v)\,\xi\,)\, \sigma$ and $h=(u\parallel_{a} v)\,\xi$; 
\item $t=u\parallel_{a} v$ and $h=t$.\\
\end{enumerate}
\end{itemize} 
\end{defi}

\noindent We now define the head reduction of a proof term: the notion generalizes Krivine's head reduction to parallel contexts. The idea is to look for the leftmost  among the head redexes of the parallel processes of a term and contract that redex. The only subtlety is to determine exactly where the new redexes for $\D$ start. Since the reduction for $u\parallel_{a} v$ is completely localized either in $u$ or $v$, it is reasonable to say that the redex starts where  the subterm $a\,u\, \sigma$ to be replaced is located. 
\begin{defi}[Letfmost Redex, Head Reduction]\label{definition-headred}\mbox{}
\begin{enumerate}
\item
The \textbf{starting symbol} of a redex $r$ is the symbol ``$($'' when $r=(u\,\xi)$ for some stack $\xi$ of length $1$; 
 it  is the leftmost occurrence of the symbol ``$a$'' such that  $a\,t\,\sigma$ is an elementary process  of $r$, when  $r=(u\parallel_{a} v)$. The \textbf{leftmost redex} among some redexes of a term $t$ is the redex whose starting symbol is the leftmost in $t$ among the starting symbols of those redexes.
\item
We say that a term $t$ \textbf{head reduces} to $t'$ and write 
$$t\succ t'$$
when $t'$ is obtained from $t$ by contracting the leftmost among the head redexes  of the parallel processes of $t$,  using one of the basic reductions in Figure \ref{fig:red}.\\
\end{enumerate}
\end{defi}
\comment{
\begin{defi}[Redex, Leftmost Redex Reduction]\mbox{}
\begin{itemize}
\item
A \textbf{redex} is any proof term $r$ such that  $r\mapsto r'$ for some $r'$ and basic reduction of Figure \ref{fig:red}. With complete parentheses, each redex is either of the form $(u\, \sigma)$ or $(u\parallel_{a} v)$. In the first case, we call the first symbol ``$($'' \textbf{active}, in the second case, the first symbol ``$($'' is \textbf{active} if $u\parallel_{a} v\mapsto u'\parallel_{a} v$, while the last symbol ``$)$'' is \textbf{active} if $u\parallel_{a} v\mapsto u\parallel_{a} v'$. The \textbf{leftmost redex} of a term $t$ is the redex of $t$, if any, whose active symbol is the leftmost in $t$ among the active symbols of all redexes of $t$.

\item We say that a term $t'$ is obtained from $t$ by \textbf{leftmost redex reduction} and write 
$$t\succ t'$$
if $t'$ is the result of replacing the leftmost redex $r$ of $t$ with the term $r'$ such that $r\mapsto r'$.  With $\mapsto^{*}$ we shall denote the reflexive and transitive closure of the one-step reduction $\mapsto$.\\
\end{itemize}
\end{defi}}\medskip

\noindent For readability, parentheses are often omitted, but in order to spot the head redex of a term, one must mentally restore the parentheses that have been suppressed.
In order to train our eye, we consider three examples of head reduction:
$$(\lambda x\, a\,(\lambda z\, z)\, x)\, u\,\parallel_{a} z_{0} \succ\,  a\, (\lambda z\, z)\, u\parallel_{a} z_{0} \succ z_{0}\parallel_{a} z_{0} $$
$$(\lambda x\, \inj_{o}(x))\, u\, [x_{0}.t_{1}, x_{1}.t_{1}]\, \sigma\, \succ\, \inj_{o}(u)\, [x_{0}.t_{0}, x_{1}.t_{1}]\,\sigma\, \succ\, t_{0}[u/x_{0}]\,\sigma$$
$$a\, \left(\left(\lambda x\, x\right) z_{0}\right) \parallel_{a} a\, z_{1}\, \succ\, \left(\lambda y\,\left( \lambda x\, x\right) z_{0}\right)\, z_{1} \parallel_{a} a\, z_{1}\, \succ\,   \left( \lambda x\, x\right) z_{0} \parallel_{a} a\, z_{1}\, \succ\, z_{0} \parallel_{a} a\, z_{1}\, \succ\,z_{0} \parallel_{a} z_{0}$$
 In the first case, the reduction for $\D$ is used as third step of the head reduction, while in the third case, as first and last step.


We define the concept of normal form and normalizable term in the usual way.
 
 \begin{defi}[Normal Forms and Normalizable Terms]\mbox{}
 \begin{itemize}
\item  A term $t$ is called a \textbf{head normal form} if there is no $t'$ such that $t\succ t'$. We define $\nf$ to be the set of head normal forms.
\item 
A sequence, finite or infinite, of proof terms $u_1,u_2,\ldots,u_n,\ldots$ is said to be a reduction of $t$, if $t=u_1$, and for all  $i$, $u_i\succ u_{i+1}$.
 A proof term $u$ of $\LC$ is (head) \textbf{normalizable} if there is no infinite reduction of $u$. We denote with $\sn$  the set of normalizable terms of $\LC$.
\end{itemize}
\end{defi}


\noindent  The reductions defined in Figure~\ref{fig:red} satisfy the important Subject Reduction Theorem: reduction steps at the level of proof terms preserve the type, which is to say that they correspond to logically sound transformations at the level of proofs. We first give the simple proof of the theorem, then analyze in detail its logical meaning in the next subsection. 
\begin{thm}[Subject Reduction]\label{subjectred}
If $t : C$ and $t \succ u$, then $u : C$. Moreover, all the free variables of $u$ appear among those of $t$.
\end{thm} 
\begin{proof} 
It is enough to prove the theorem for basic reductions: if $ t : C$ and $t \mapsto u$, then  $u : C$. The proof that the intuitionistic reductions and the permutation rules preserve the type is completely standard. Thus we are left with the $\D$-reductions, which require straightforward considerations as well. Suppose
 $$\mathcal{C}[a^{\scriptscriptstyle A\rightarrow B}\, u\,\sigma]\parallel_{a} v\ \mapsto\  \mathcal{C}[\,v[\lambda y^{\scriptscriptstyle B}\, u /a^{\scriptscriptstyle B\rightarrow A}]\,]\parallel_{a} v$$
Since $\mathcal{C}$ is a parallel context, $a^{\sml{A\rightarrow B}}\, u\,\sigma$ and $v$ have both type $C$. Now, $u$ must be of type $A$, so $\lambda y^{B} u$ is of type $B\rightarrow A$ and thus 
 $v[\lambda y^{\scriptscriptstyle B}\, u /a^{\scriptscriptstyle B\rightarrow A}]$ is a correct term of type $C$. Moreover, all the occurrences of $a^{\scriptscriptstyle B\rightarrow A}$ in $v$ are eliminated by the substitution $[\lambda y^{\scriptscriptstyle B}\, u /a^{\scriptscriptstyle B\rightarrow A}]$, so no new free variable is created.
 \end{proof}

\comment{
We conclude the treatment of reduction rules with another example, that also suggests that some other reduction rules, considered in \cite{AschieriZH}, might be useful. The reductions are: 

\begin{description}
\item[Optional Reduction Rules for $\D$]
\[u\parallel_{a}v \mapsto u \mbox{ if $a$ does not occur in $u$}\]
\[u\parallel_{a}v \mapsto v \mbox{ if $a$ does not occur in $u$}\]
\end{description}
Whereas in $\cite{AschieriZH}$ the first reduction is essential, in this setting both the first and the second are useless for all the main results of the paper. However, they can be used to show that the type $A\lor B$ is redundant in the logic $\LC$. This follows of course from the fact that $A\lor B$ is definable by means of $\rightarrow$ and $\land$, as noticed by Dummett \cite{Dummett}. We can define:
$$(A\lor B)^{*}:=  (A\rightarrow B)\rightarrow B \land (B\rightarrow A)\rightarrow A$$
Then, if $u: A$, we define 
$$(\iota_{0}(u))^{*}:= \langle \lambda y^{{\scriptscriptstyle A\rightarrow B}}\, y\, u, \lambda z^{{\scriptscriptstyle B\rightarrow A}} u\rangle$$
and if $u: B$, we define
$$(\iota_{1}(u))^{*}:= \langle \lambda z^{{\scriptscriptstyle A\rightarrow B}} u, \lambda y^{{\scriptscriptstyle B\rightarrow A}}\, y\, u\rangle$$
where $z$ is a dummy variable not occurring in $u$. 
Finally,  we define
$$(t\, [x_{0}^{\sml{A}}.u_{0}, x_{1}^{\sml{B}}.u_{1}])^{*}:=  (\lambda x_{1}^{\sml{B}}\, u_{1})(t\, \pi_{0}\, a^{{\scriptscriptstyle A\rightarrow B}})\parallel_{a} (\lambda x_{0}^{\sml{A}}\, u_{1})(t\, \pi_{1}\, a^{\sml{B\rightarrow A}})$$
Then
$$((\iota_{0}(u))^{*}  [x_{0}^{\sml{A}}.u_{0}, x_{1}^{\sml{B}}.u_{1}])^{*}$$
}
\subsection{Reduction Rules: Logical Interpretation}

So far, in studying the system $\LC$, we have given priority to the underlying lambda calculus and characterized it as a functional language endowed with parallelism and a communication mechanism. The explanation of the reductions had little to do with logic and much with computation. However, thanks to the Subject Reduction Theorem, we know we could have proceeded the other way around. Namely, we could have given priority to logic and dealt only with transformation of proofs, in the style of Prawitz natural deduction trees \cite{Prawitz}. Since it is instructive to explain directly this point of view, we are finally going to do so. 

First of all, the following proof of $\lnot A\lor \lnot\lnot A$ is an example of natural deduction tree in $\LC$:
\begin{prooftree}

\AxiomC{$[\lnot\lnot A\rightarrow \lnot A]$}

\AxiomC{$[\lnot A]$}
\AxiomC{$[A]$}
\BinaryInfC{$\bot$}
\UnaryInfC{$\lnot \lnot A$}
\BinaryInfC{$\lnot A$}
\AxiomC{$	[ A]$}
\BinaryInfC{$\bot$}
\UnaryInfC{$\lnot A$}
\UnaryInfC{$\lnot A\lor \lnot\lnot A$}

\AxiomC{$[\lnot A\rightarrow \lnot\lnot A]$}
\AxiomC{$[\lnot A]$}
\BinaryInfC{$\lnot\lnot A$}
\AxiomC{$	[\lnot A]$}
\BinaryInfC{$\bot$}
\UnaryInfC{$\lnot\lnot A$}
\UnaryInfC{$\lnot A\lor \lnot\lnot A$}
\RightLabel{$\D$}
\BinaryInfC{$\lnot A\lor \lnot\lnot A$}
\end{prooftree}
The standard reductions for lambda calculus still correspond to the ordinary conversions for all the logical constants of first-order logic:

$$\begin{aligned}
 \AxiomC{$[A]$}
\noLine
\UnaryInfC{$\vdots$}
\noLine
\UnaryInfC{$B$}
\UnaryInfC{$A\rightarrow B$}
\AxiomC{$\vdots$}
\noLine
\UnaryInfC{$A$}
\BinaryInfC{$B$}
\DisplayProof
  &
  \qquad \mbox{ \textsf{converts to:} } \qquad
  &
 \AxiomC{$\vdots$}
\noLine
\UnaryInfC{$A$}
\noLine
\UnaryInfC{$\vdots$}
\noLine
\UnaryInfC{$B$}
\DisplayProof
 \\
\end{aligned}
$$

$$
\begin{aligned}
\AxiomC{$\vdots$}
\noLine
\UnaryInfC{$A_i$}
\RightLabel{{\scriptsize$(i\in\{1, 2\})$}}
\UnaryInfC{$A_1\lor A_2$}
\AxiomC{$[A_1]$}
\noLine
\UnaryInfC{$\vdots$}
\noLine
\UnaryInfC{$C$}
\AxiomC{$[A_2]$}
\noLine
\UnaryInfC{$\vdots$}
\noLine
\UnaryInfC{$C$}
\TrinaryInfC{$C$}
\DisplayProof
& \qquad\mbox{ \textsf{converts to}: }
&
\AxiomC{$\vdots$}
\noLine
\UnaryInfC{$A_i$}
\noLine
\UnaryInfC{$\vdots$}
\noLine
\UnaryInfC{$C$}
\DisplayProof\\
\end{aligned}
$$

$$\begin{aligned}
 \AxiomC{}
\noLine
\UnaryInfC{$\vdots$}
\noLine
\UnaryInfC{$A_{1}$}
\AxiomC{}
\noLine
\UnaryInfC{$\vdots$}
\noLine
\UnaryInfC{$A_{2}$}
\BinaryInfC{$A_{1}\land A_{2}$}
\RightLabel{{\scriptsize$(i\in\{1, 2\})$}}
\UnaryInfC{$A_{i}$}
\DisplayProof
  &
  \qquad \mbox{ \textsf{converts to:} } \qquad
  &
\AxiomC{}
\noLine
\UnaryInfC{$\vdots$}
\noLine
\UnaryInfC{$A_{i}$}
\DisplayProof
 \\
\end{aligned}
$$

$$\begin{aligned}
\AxiomC{$$}
\noLine
\UnaryInfC{$\vdots$}
\noLine
\UnaryInfC{$A[m/\alpha]$}
\UnaryInfC{$\exists \alpha\, A$}
\AxiomC{$[A]$}
\noLine
\UnaryInfC{$$}
\noLine
\UnaryInfC{$\pi$}
\noLine
\UnaryInfC{$$}
\noLine
\UnaryInfC{$C$}
\BinaryInfC{$C$}
\DisplayProof
  &
  \qquad \mbox{ \textsf{converts to:} } \qquad
  &
 \AxiomC{$\vdots$} 
\noLine
\UnaryInfC{$A[m/\alpha]$}
\noLine
\UnaryInfC{$$}
\noLine
\UnaryInfC{$\pi[m/\alpha]$}
\noLine
\UnaryInfC{$$}
\noLine
\UnaryInfC{$C$}
\DisplayProof
 \\
\end{aligned}
$$

$$\begin{aligned}
\AxiomC{$\pi$}
\noLine
\UnaryInfC{$$}
\noLine
\UnaryInfC{$A$}
\UnaryInfC{$\forall \alpha\, A$}
\UnaryInfC{$A[m/\alpha]$}
\DisplayProof
  &
  \qquad \mbox{ \textsf{converts to:} } \qquad
  &
\AxiomC{$\pi[m/\alpha]$}
\noLine
\UnaryInfC{$$}
\noLine
\UnaryInfC{$A[m/\alpha]$}
\DisplayProof
 \\
\end{aligned}
$$\medskip

\noindent The permutation reductions for the terms of the form $u\parallel_{a}{v}$, are just instances of Prawitz-style permutations for disjunction elimination. From the logical perspective, they are used to systematically transform, whenever possible, the logical shape of the conclusion. This reduction is essential because the Dummett inference rule does not yield much when employed to prove implications  or disjunctions; but it becomes Herbrand constructive, whenever used to prove existentially quantified statements. As an example of permutation for $\D$, we consider the one featuring an implication as conclusion:

$$\begin{aligned}
\AxiomC{$[A\rightarrow B]$}
\noLine
\UnaryInfC{$\vdots$}
\noLine
\UnaryInfC{$F \to G$}
\AxiomC{$[B\rightarrow A]$}
\noLine
\UnaryInfC{$\vdots$}
\noLine
\UnaryInfC{$F \to G$}
\RightLabel{$\D$}
\BinaryInfC{$F\to G$}
\AxiomC{$\vdots$}
\noLine
\UnaryInfC{$F$}
\BinaryInfC{$G$}
\DisplayProof
  &
  \ \ \ \mbox{ \textsf{converts to:} } \ \ \
  &
\AxiomC{$[A\rightarrow B]$}
\noLine
\UnaryInfC{$\vdots$}
\noLine
\UnaryInfC{$F \to G$}
\AxiomC{$\vdots$}
\noLine
\UnaryInfC{$F$}
\BinaryInfC{$G$}
\AxiomC{$[B\rightarrow A]$}
\noLine
\UnaryInfC{$\vdots$}
\noLine
\UnaryInfC{$F \to G$}
\AxiomC{$\vdots$}
\noLine
\UnaryInfC{$F$}
\BinaryInfC{$G$}
\BinaryInfC{$G$}
\DisplayProof
 \\
\end{aligned}
$$
There are similar permutations for all other elimination rules, as one can see translating in natural deduction the permutations of Figure \ref{fig:red}.
With the following notation
\begin{prooftree}
\AxiomC{$\mathcal{D}_{1}$}
\noLine
\UnaryInfC{$C$}
\AxiomC{$\cdots$}
\AxiomC{$\mathcal{D}_{i}$}
\noLine
\UnaryInfC{$C$}
\AxiomC{$\cdots$}
\AxiomC{$\mathcal{D}_{n}$}
\noLine
\UnaryInfC{$C$}
\doubleLine
\RightLabel{$\D$}
\insertBetweenHyps{\hskip 4pt}
\QuinaryInfC{$C$}
\end{prooftree}
we denote a deduction of $C$ that, in order to obtain its final conclusion,
 combines the deductions $\mathcal{D}_{1}, \ldots, \mathcal{D}_{i}, \ldots, \mathcal{D}_{n}$ of $C$ using \emph{only} the Dummett rule $n-1$ times. In other words, below the conclusions $C$ of the deductions $\mathcal{D}_{1}, \ldots, \mathcal{D}_{i}, \ldots, \mathcal{D}_{n}$ only the Dummett rule is used. This configuration corresponds to a parallel context in our lambda calculus, as in Definition \ref{defi-parallelc}. With the  notation
\begin{prooftree}
\AxiomC{$B$}
\doubleLine\RightLabel{$\mathsf{EL}$}
\UnaryInfC{$C$}
\end{prooftree}
we denote a deduction of $C$ that, starting from $B$, applies only elimination rules to obtain $C$; in particular, $B$ must be the main premise of the first elimination rule which concludes $B_{1}$, which must be the main premise of the second elimination rule which concludes $B_{2}$ and so on down to $C$. 
This configuration corresponds to the concept of stack of Definition \ref{definition-stack}.

Finally, we can look at the two reductions for proofs containing the Dummett rule. 
Let us consider just the first conversion for $\D$, the second being perfectly symmetric:

{\scriptsize
$$\begin{aligned}
\AxiomC{$\mathcal{D}_{1}$}
\noLine
\UnaryInfC{$C$}
\AxiomC{$\cdots$}
\AxiomC{$[A\rightarrow B]$}
\AxiomC{$[A\rightarrow B]$}
\noLine
\UnaryInfC{$\vdots$}
\noLine
\UnaryInfC{$A$}
\insertBetweenHyps{\hskip -3pt}
\BinaryInfC{$B$}
\doubleLine\RightLabel{$\mathsf{EL}$}
\UnaryInfC{$C$}
\AxiomC{$\cdots$}
\AxiomC{$\mathcal{D}_{n}$}
\noLine
\UnaryInfC{$C$}
\doubleLine
\RightLabel{$\D$}
\insertBetweenHyps{\hskip -2pt}
\QuinaryInfC{$C$}
\noLine
\AxiomC{$[B\rightarrow A]$}
\UnaryInfC{$\mathcal{D}$}
\noLine
\UnaryInfC{$C$}
\RightLabel{$\D$}
\BinaryInfC{$C$}
\DisplayProof
&
  \mbox{\textsf{converts to:}} 
 &
 \AxiomC{$\mathcal{D}_{1}$}
\noLine
\UnaryInfC{$C$}
\AxiomC{$\cdots$}
\AxiomC{$[A\rightarrow B]$}
\noLine
\UnaryInfC{$\vdots$}
\noLine
\UnaryInfC{{$A$}}
\UnaryInfC{$B\rightarrow A$}
\noLine
\UnaryInfC{$\mathcal{D}$}
\noLine
\UnaryInfC{$C$}
\AxiomC{$\cdots$}
\AxiomC{$\mathcal{D}_{n}$}
\noLine
\UnaryInfC{$C$}
\doubleLine
\RightLabel{$\D$}
\insertBetweenHyps{\hskip -2pt}
\QuinaryInfC{$C$}
\noLine
\AxiomC{$[B\rightarrow A]$}
\UnaryInfC{$\mathcal{D}$}
\noLine
\UnaryInfC{$C$}
\RightLabel{$\D$}
\BinaryInfC{$C$}
\DisplayProof
\end{aligned}$$}

\comment{\scriptsize
$$\begin{aligned}
\AxiomC{$[A\rightarrow B]$}
\AxiomC{$[A\rightarrow B]$}
\noLine
\UnaryInfC{$\vdots$}
\noLine
\UnaryInfC{$A$}
\insertBetweenHyps{\hskip -3pt}
\BinaryInfC{$B$}
\doubleLine
\RightLabel{$\mathsf{EL}$}
\UnaryInfC{$C$}
\noLine
\UnaryInfC{$\vdots$}
\noLine
\UnaryInfC{$C$}
\noLine
\AxiomC{$[B\rightarrow A]$}
\UnaryInfC{$\vdots$}
\noLine
\UnaryInfC{$C$}
\RightLabel{$\D$}
\BinaryInfC{$C$}
\DisplayProof
&
  \mbox{\textsf{converts to:}} 
 &
 \AxiomC{$[A\rightarrow B]$}
\noLine
\UnaryInfC{$\vdots$}
\noLine
\UnaryInfC{{$A$}}
\UnaryInfC{$B\rightarrow A$}
\noLine
\UnaryInfC{$\vdots$}
\noLine
\UnaryInfC{$C$}
\noLine
\UnaryInfC{$\vdots$}
\noLine
\UnaryInfC{$C$}
\noLine
\AxiomC{$[B\rightarrow A]$}
\UnaryInfC{$\vdots$}
\noLine
\UnaryInfC{$C$}
\RightLabel{$\D$}
\BinaryInfC{$C$}
\DisplayProof
\end{aligned}$$}
\noindent The conversion above focuses first on the deduction $\mathcal{D}$ on the left branch of the proof; it replaces the hypothesis $B\rightarrow A$ of $\mathcal{D}$ with a proof of $B\rightarrow A$ directly obtained from the proof of $A$ found on the left branch; afterwards, it takes the deduction so generated and replaces with it the old proof of $C$ obtained from $B$ by elimination rules.\\
There is a  crucial assumption about the structure of the first proof. In the left branch of the Dummett rule, the hypothesis $A\rightarrow B$ is used together with $A$ to obtain $B$, which is in turn used  to infer $C$ \emph{by means only of a main branch of elimination rules}, as called by Prawitz.  
 Thanks to this restriction, the proof of $A$ does not end up having more open assumptions in the second proof than it has in the first proof.\\
But what have we gained with this reduction? It looks like we made no progress at all. The hypothesis $A\rightarrow B$ may be actually used more times in the second proof than in the first, because the hypothesis $B\rightarrow A$ might be used several times in the deduction $\mathcal{D}$! Actually, the gain is subtle. In the left branch of the first proof the formula $B$ was derived in a fictitious way: by an arbitrary hypothesis $A\rightarrow B$, bearing no relationship with $C$. Since $B$ is \emph{used} to obtain $C$, we cannot expect $B$ to provide constructive content to $C$, in particular no witness if $C$ is an existential formula.  The conversion above gets rid of this configuration and provide a more direct proof of $C$: in the new proof, if $B\rightarrow A$ is employed to derive $A$ by modus ponens, one can discard $B$ and use the proof of $A$ coming from the first proof.

 The main difficulty that we face with our reduction rules for $\D$ is \emph{termination}. There is hardly any decrease in complexity from before to after the reduction and the road toward a combinatorial termination proof looks barred. We are thus forced to employ a far more abstract technique: realizability.

\section{Classical Realizability}
\label{section-systemNEM}

In this section we prove that each term of  $\LC$ \emph{realizes} its type and is normalizing. To this end, we make a detour into a logically inconsistent, yet computationally sound world: the system $\ALC$, a type system which extends $\LC$ .
The idea that extending a system can make easier rather than harder to prove its normalization might not seem very intuitive, but it is well tested and very successful (see \cite{Tait}, \cite{rattiTLCA}, \cite{rattoCOS}, \cite{AschieriZH}). $\ALC$ will be our calculus of realizers. It is indeed typical of realizability, the method we shall use, to set up a calculus with more realizers than the actual proof terms \cite{Kreiselm, Krivine, AZMarkov}. The idea is that a realizer is defined as a proof term that defeats every opposer and passes every termination test; but proof terms, as opposers and testers of proof terms themselves, are not enough; proof terms must be opposed and tested also by ``cheaters'', terms that do satisfy the same definition of realizability, but only because they have some advantage. These extra tests make proof terms stronger realizers than they otherwise would be. We may imagine a realizer as a tennis player that trains himself to return fast balls thrown by a robot: if he withstands the attacks of the robot, he will perform all the more well against real weaker humans.

\subsection{The Abort Operator}

The system $\ALC$  is not meant to be a logical system: it would be inconsistent! The purpose of the system is not logical, but computational: to simulate the reduction rules for $\D$ by an abort operator $\abort$.
We define the typing rules of  $\ALC$ to be those of $\LC$ plus a new term formation scheme:\\
\begin{description}
\item[Abort Axiom] 
\AxiomC{$\abort^{\scriptscriptstyle A\rightarrow B}: A\rightarrow B$}
\DisplayProof
\end{description}
\vskip 0.15in
With $\abort$, $\abort_{1}, \ldots,\abort_{k}$, we shall denote some generic constant $\abort^{\scriptscriptstyle A\rightarrow B}$. The reduction rules for the terms of $\ALC$ are those for $\LC$ with the addition of a new reduction rule defined in Figure \ref{fig:FN}.
\begin{figure*}[!htb]
\footnotesize{
\dlinea
\begin{description}
\item[Reduction Rules for $\abort$]
\[\abort\, u\, \sigma \mapsto u\]
\[\mbox{whenever $\abort\, u\, \sigma$ and $u$ have the same type }\]

\end{description}}
\dlinea
\caption{Extra Reduction Rules for $\ALC$}\label{fig:FN}
\end{figure*}

\comment{
\begin{figure*}
 \begin{tabular}{| c |}
      \hline
[Reduction Rules for $\IL$]
$(\lambda x. u)t\redn u[t/x]\qquad (\lambda \alpha. u)t\redn u[t/\alpha]$\\
$ \pair{u_0}{u_1}\pi_{i}\redn u_i, \mbox{ for $i=0,1$}$\\
$\inj_{i}(u)[x_{1}.t_{1}, x_{2}.t_{2}]\redn t_{i}[u/x_{i}], \mbox{ for $i=0,1$}$\\
$(n, u)[(\alpha,x).v]\redn v[n/\alpha][u/x], \mbox{ for each numeral $n$}$\\
\hline
    \end{tabular}
\caption{New Reduction Rules for $\IL$ + $\AD$}\label{fig:FN}
\end{figure*}
}
The abort computational construct reminds Krivine's  $\mathsf{k}_{\pi}$, which removes the current continuation $\rho$ and restore a previously saved continuation $\pi$:
$$\mathsf{k}_{\pi}\star t\dott \rho \succ t\star \pi$$ 
There is indeed an analogy with Krivine's realizability: the terms of $\LC$ correspond to Krivine's proof-like terms, whereas the terms of $\ALC$ correspond to Krivine's inconsistent terms that may contain $\mathsf{k}_{\pi}$ and may realize any formula. But in our case $\ALC$ is just a \emph{tool for defining realizability}, not a tool for implementing reductions, like $\mathsf{k}_{\pi}$ in Krivine case.
The role of $\abort$ will emerge later on in the proof of Propositions \ref{proposition-simul} and \ref{proposition-parallel}. However, by now, the intuition should be pretty clear: in the reduction
$$\mathcal{C}[a\, u\,\sigma]\parallel_{a} v\ \mapsto\  \mathcal{C}[\,v[\lambda y\, u /a]\,]\parallel_{a} v$$
the term $a\, u\,\sigma$ aborts the local continuation $\sigma$. The difficulty is that the new continuation $v[\lambda y\, u /a]$, from the perspective of $a$, is created out of nowhere! Therefore proving by induction that $\mathcal{C}[a\, u\,\sigma]$ is a realizer would not be of great help for proving that the whole term $\mathcal{C}[a\, u\,\sigma]\parallel_{a} v$ is realizer. With terms of the form $\abort\, {w}$ we can instead simulate locally the global reduction above and get a stronger induction hypothesis.

The Definition \ref{definition-stack} of stack is of course extended to $\ALC$ and the Definition \ref{definition-head} of head redex is extended to the terms of $\ALC$ by saying that $\abort\, u\,\sigma$ is the \textbf{head redex} of $\abort\, u\,\sigma$ whenever $u$ and $\abort\, u\,\sigma$ have the same type. The reduction relation $\succ$ for the terms of $\ALC$ is then defined as in Definition \ref{definition-headred}.
\comment{
\begin{defi}[Head Reduction for $\ALC$]\mbox{}
\begin{itemize}
\item We say that a term $t'$ is obtained from $t$ by \textbf{leftmost redex reduction} and write 
$$t\redn t'$$
if $t'$ is the result of replacing the leftmost redex $r$ of $t$ with the term $r'$ such that $r\mapsto r'$.  With $\redn^{*}$ we shall denote the reflexive and transitive closure of the one-step reduction $\redn$ and with $\redn^{+}$ the transitive closure.\\
\end{itemize}
\end{defi}}
In the following,  we define $\snn$ to be the set of  normalizing proof terms of $\ALC$. 

As usual in lambda calculus, a value represents the result of the computation: a function for arrow and universal types, a pair for product types, a boolean for sum types and a witness for existential types and in our case also the abort operator.

\begin{defi}[Values, Neutrality]\label{definition-value}\mbox{}
\begin{itemize}
\item A proof term  is a \textbf{value} if it is of the form $\lambda x\, u$ or $\lambda \alpha\, u$ or $\pair{u}{t}$ or $\inj_{i}(u)$ or $(m, u)$ or $\efq{}{u}$ or $\abort$.
\item A proof term is  \textbf{neutral} if it is neither a value nor of the form $u\parallel_{a} v$.
\end{itemize}
\end{defi}

We now prove a property of head normal forms that we will be crucial in the following. It is a generalization of the well known head normal form Theorem for lambda calculus and tells us that if we decompose a proof term into its elementary parallel processes, then each of them is either a value or some variable or constant applied to a list of argument.
\begin{prop}[Head Normal Form Property]\label{proposition-hnf}
Suppose $t$ is in head normal form and 
$$t=t_{1}\parallel_{a_{1}} t_{2}\parallel_{a_{2}}\ldots \parallel_{a_{n}} t_{n+1}$$
and that each $t_{i}$ is an elementary process.
 Then for every $1 \leq i\leq n+1$, there is some stack $\sigma$ such that either $t_{i}=x\, \sigma$, with $x\neq a_{1},\ldots, a_{n}$, or $t_{i}=\abort\, u\, \sigma$, with the type of $u$ different from the type of $\abort\, u\, \sigma$, or  $t_{i}=a_{j}$  or $t_{i}$ is a value.
\end{prop}
\begin{proof} By induction on $t$. If $t$ is a value, we are done. There are two other cases to consider.
\begin{enumerate}
\item $t=u\parallel_{a} v$. By Definition \ref{definition-head}, the parallel processes of $u$ and $v$ are parallel processes of $u\parallel_{a}v$ as well, so they cannot have head redexes; hence $u$ and $v$ are in head normal form. By induction hypothesis, $u$ and $v$ are of the desired form, thus we just have to check that if
$$u=u_{1}\parallel_{a_{1}} u_{2}\parallel_{a_{2}}\ldots \parallel_{a_{k}} u_{k+1}$$
and for some $i$,  $u_{i}=x\,\sigma$, with $\sigma\neq \epsilon$, then $x\neq a$ (and symmetricallly for $v$). Indeed, if for some $i$, $u_{i}=a\,\sigma$, then $u=\mathcal{C}[a\,\sigma]$ for some parallel context $\mathcal{C}[\ ]$, and therefore $u\parallel_{a}v$ would be the leftmost head redex of itself, which is impossible since by assumption it is in head normal form.
\item $t$ is neutral. Then $t$ can be written, for some stack $\sigma$, as $r\, \sigma$ where $r$ is a value or $r=u\parallel_{a}v$ or $r=x$. In the third case, we are done; in the first and second case, $\sigma=\xi\dott \rho$, so $r\,\xi$ would be the head redex of $t$, unless $t=\abort\,\xi\,\rho$, with the type of $\xi$ different from the type of $\abort\,\xi\,\rho$, which is the thesis.\qedhere
\end{enumerate}
\end{proof}
\subsection{Definition of Classical Realizability}\label{section-reducibility}

Our main goal now is to prove the Normalization Theorem for $\LC$: every proof term of $\LC$ reduces in a finite number of head reduction steps to a head normal form. We shall employ a notion of \emph{classical realizability}, a generalization of the Tait-Girard reducibility method \cite{Girard} that works for classical type systems. The origins of classical realizability can be traced all the way back to Parigot \cite{Parigot} and Krivine \cite{Krivine0} classical reducibility, but we present it in a fashion popularized later by Krivine in his work on realizability \cite{Krivine}, which is indeed a generalization of classical reducibility. Thanks to the fact that one considers only head reduction, Krivine-style classical realizability is slightly simpler than the notions usually employed to derive strong normalization. 

Given a logic, we raise a question: what kind of evidence does a proof provide other than the tiny bit ``1'' declaring the truth of the proven statement? Realizability is a semantics explaining what is to be taken as \emph{constructive evidence for a statement} and a technique for  showing  that proofs can provide such an evidence. Formally, realizability is a relation between terms of $\ALC$ and formulas, with terms playing the role of constructions and formulas determining what properties a construction should satisfy. 
In particular, to each formula $C$ is associated a set of stacks $||C||$, which represents a collection of \emph{valid tests}:  whenever a term passes all these tests, in the sense that it maps them into terminating programs, it is a realizer.  As prescribed by the pragmatist viewpoint, the clauses that defines realizability follow the shape of elimination rules, in order to make sure that no matter how a program is used, it always terminates. 



\begin{defi}[Valid Tests, Classical Realizability]
\label{definition-reducibility}
Assume $t$ is a term of $\ALC$ and $C$ is a formula of $\Language$. We define by mutual induction the relation $t\red C$ (``$t$ realizes $C$'')  and a set $|| C ||$ of stacks of $\ALC$ (the ``valid tests for $C$'') according to the form of $C$:
\begin{itemize}

\item $t\red C$ if and only if $t: C$ and for all $\sigma\in ||C||$, $t\,\sigma \in \snn$

\item
$||\emp{}|| = \{\epsilon\}$


\item
$||A\rightarrow B||=\{u\dott\sigma\ |\  u\red A \land \sigma\in ||B||\}\cup \{\epsilon\}$

\item
$||A\land B||=\{\pi_0\dott \sigma\ |\  \sigma\in ||A||\}\cup \{\pi_1\dott \sigma\ |\  \sigma \in ||B||\}\cup \{\epsilon\}$ 

\item
$||A\lor B||=\{ [x.u, y.v]\dott\sigma \ |\ \forall t.\ (t \red A\implies u[t/x]\,\sigma \in \snn) \land  (t\red B\implies v[t/y]\, \sigma \in \snn)\}\cup \{\epsilon\}$ 

\item
$||\forall \alpha\, A||=\{m\dott\sigma\ |\ m\in\Language \land \sigma \in ||A[{m}/\alpha]||\}\cup \{\epsilon\}$
\item

$||\exists \alpha\,A||=\{ [(\alpha, x).v]\dott\sigma \ |\ \forall t.\ t \red A[m/\alpha]\implies v[m/\alpha][t/x]\, \sigma \in \snn \}\cup \{\epsilon\}$
\end{itemize}
\end{defi}

\subsection{Properties of Realizers}\label{section-reducibilityproperties}

In this section we prove the basic properties of classical realizability. They are all we need to prove the Adequacy Theorem \ref{Adequacy Theorem}, which states that typable terms are realizable.
The arguments for establishing the properties are in many cases standard (see Krivine \cite{Krivine}). We shall need extra work for dealing with terms of the form $u\parallel_{a} v$.

The first task is to prove that realizability is sound for all introduction and elimination rules of $\LC$. We start with the eliminations.

\begin{prop}[Properties of Realizability: Eliminations]\label{proposition-redelim}\mbox{}
\begin{enumerate}
\item If $t\red A\rightarrow B$ and $u\red A$, then  $tu \red B$.
\item If $t\red \forall \alpha\, A$, then for every term $m$ of $\Language$, $t m\red A[m/\alpha]$.
\item If $t\red A \land B$, then $t\, \pi_{0} \red  A$ and $t\, \pi_{1} \red  B$.
\item If $t\red A\lor B$ and for every $w\red A$, $u[w/x]\red C$ and for every $w\red B$, $v[w/y]\red C$, then $t\, [x.u, y.v] \red C$.
\item If  $t\red \exists \alpha\, A$ and for every $m\in \Language$ and  for every $w\red A[m/\alpha]$, $u[m/\alpha][w/x]\red C$, then $t\, [(\alpha, x).u]\red C$.
\end{enumerate}
\end{prop}

\begin{proof}\mbox{}
\begin{enumerate}
\item Assume $t\red A\rightarrow B$ and $u\red A$. Let $\sigma \in ||B||$; we must show $t u\,\sigma\in\snn$. Indeed, since $t\red A$, by Definition \ref{definition-reducibility} $u\dott \sigma \in ||A\rightarrow B||$ and since $t\red A\rightarrow B$,  we conclude $t u\,\sigma\in\snn$.
\item Similar to 1.
\item Assume $t\red A \land B$. Let $\sigma\in||A||$; we must show $t\, \pi_{0}\,\sigma \in\snn$. Indeed, by Definition \ref{definition-reducibility} $\pi_{0}\dott \sigma \in ||A\land B||$ and since $t\red A \land B$, we conclude $t\, \pi_{0}\, \sigma\in\snn$. A symmetrical reasoning shows that $t\, \pi_{1}\red B$.
\item Let $\sigma\in ||C||$. We must show that $t\, [x.u,y.v]\,\sigma\in\snn$. By hypothesis, for every $w\red A$, $u[w/x]\,\sigma \in\snn$ and for every $w\red B$, $v[w/y]\,\sigma  \in\snn$; by Definition \ref{definition-reducibility}, $[x.u, y.v]\,\sigma\in ||A\lor B||$. Since $t\red A\lor B$, we conclude $t\, [x.u, y.v]\,\sigma\in\snn $. 
\item Similar to 4.\qedhere
\end{enumerate}
\end{proof}
Realizability is also sound for introduction rules and the abort operator realizes any implication.
\begin{prop}[Properties of Realizability: Introductions]\label{proposition-somecases}\mbox{}
\begin{enumerate}
\item If for every $t\red A$, $u[t/x]\red B$, then  $\lambda x\, u\red A\rightarrow B$.
\item If for every term $m$ of $\Language$, $u[m/\alpha]\red B[m/\alpha]$, then $\lambda \alpha\, u\red \forall \alpha\, B$.
\item If $u\red A$ and $v\red B$, then $\pair{u}{v}\red  A\land B$.
\item If $t\red A_{i}$, with $i\in\{0,1\}$, then $\inj_{i}(t)\red A_{0}\lor A_{1}$.
\item If  $t\red A[m/\alpha]$, then $(m,t)\red \exists \alpha\, A$.
\item If  $A$ and $B$ are any two formulas, then $\abort  \red A\rightarrow B$.
\end{enumerate}
\end{prop}
\begin{proof}\mbox{}
\begin{enumerate}
\item Suppose that for every $t\red A$, $u[t/x]\red B$. Let $\sigma \in ||A\rightarrow B||$. We have to show $(\lambda x\, u)\,\sigma\in\snn$.  If $\sigma=\epsilon$, indeed $\lambda x\, u\in\snn$.  Suppose then $\sigma= t\dott \rho$, with $t \red A$ and $\rho\in ||B||$.  Since by hypothesis $u[t/x]\red B$, we have $u[t/x]\,\rho\in \snn$; moreover, $(\lambda x\, u) t\, \rho \redn u[t/x]\, \rho$. Therefore, $(\lambda x\, u) t\, \rho\in\snn$.
\item Similar to 1.
\item Suppose $u\red A$ and $v\red B$. Let $\sigma\in ||A\land B||$. We have to show $\pair{u}{v}\,\sigma\in\snn$. If $\sigma=\epsilon$, indeed $\pair{u}{v}\in\snn$. Suppose then $\sigma=\pi_{i}\dott\rho$, with $i\in\{0,1\}$ and $\rho \in ||A||$, when $i=0$,  and $\rho \in ||B||$, when $i=1$. We have two cases.
\begin{enumerate}
\item $i=0$. Since by hypothesis $u\red A$, we have $u\,\rho\in \snn$; moreover, $\pair{u}{v}\,\pi_{i}\,\rho \redn u\, \rho$. Therefore, $\pair{u}{v}\,\pi_{i}\,\rho\in\snn$.
\item $i=1$. Since by hypothesis $u\red B$, we have $v\,\rho\in \snn$; moreover, $\pair{u}{v}\,\pi_{i}\,\rho \redn v\, \rho$. Therefore, $\pair{u}{v}\,\pi_{i}\,\rho\in\snn$.
\end{enumerate}

\item Suppose $t\red A_{i}$. Let $\sigma\in ||A_{0}\lor A_{1}||$. We have to show $\inj_{i}(t)\,\sigma\in\snn$. If $\sigma=\epsilon$, indeed $\inj_{i}(t)\in\snn$. Suppose then $\sigma=[x_{0}.u_{0}, x_{1}.u_{1}]\dott \rho$ and for all $w$, if $w\red{A_{0}}$, then $u_{0}[w/x_{0}]\,\rho\in \snn$ and if $w\red{A_{1}}$, then $u_{1}[w/x_{1}]\, \rho\in \snn$. We have to show $\inj_{i}(t)\, [x_{0}.u_{0}, x_{1}.u_{1}]\,\rho \in \snn$. By hypothesis $u_{i}[t/x_{i}]\,\rho\in \snn$; moreover, $$\inj_{i}(t)\, [x_{0}.u_{0}, x_{1}.u_{1}]\,\rho \redn u_{i}[t/x_{i}]\, \rho$$ Therefore, $\inj_{i}(t)\, [x_{0}.u_{0}, x_{1}.u_{1}] \in\snn$.
 
\item Similar to 4. 
\item Let $\sigma\in ||A\rightarrow B||$. We have to show that $\abort\,\sigma\in\snn$.  If $\sigma=\epsilon$, indeed $\abort\in\snn$. Suppose then $\sigma=u\dott \rho$, with $u\red A$ and $\rho\in ||B||$. We have to show that $\abort\,u\,\rho\in\snn$. Since $u\red A$ and $\epsilon \in ||A||$, we have $u=u\,\epsilon\in\snn$. Moreover, if $\abort\,u\,\rho$ is not a redex, we are done, and if $\abort\,u\,\rho \redn u$,  the thesis follows.\qedhere
\end{enumerate}
\end{proof}\smallskip

\noindent It is now that the abort operator really enters the scene.  Thanks to it, any reduction $u\parallel_{a} v\redn u'\parallel_{a} v$ can be simulated in a purely local way. This is possible because any such reduction affects only what is inside $u$ and leaves $v$ untouched. Then, in order to replicate the reduction is enough to substitute to $a$ a term $\infi{A}$ that throws away any stack of terms it is applied to, like $a$ does, and then restores $v$, with some substitution depending on the context. Of course, symmetrical considerations hold true for any reduction  $u\parallel_{a} v\redn u\parallel_{a} v'$.

\begin{prop}[Local Simulation]\label{proposition-simul}
Define 
$$\infi{A}:= \lambda x\, \abort\, {v[\lambda y\, x /a]}$$
$$\infi{B}:= \lambda z\, \abort\, { u[\lambda y\, z /a]}$$
with $x$, $y$, $z$ and $\abort$ occurring with the right type. Then
$$(u\parallel_{a} v\redn u'\parallel_{a} v)\implies u[\infi{A}/a]\redn^{+} u'[\infi{A}/a]$$
$$(u\parallel_{a} v\redn u\parallel_{a} v')\implies v[\infi{B}/a]\redn^{+} v'[\infi{B}/a]$$
\end{prop}
\proof
We prove the first statement, the other being perfectly symmetric. The only trouble is to formalize precisely the argument, which is otherwise intuitively obvious. To this end, we first need some simple, but tedious to prove, claims. 
\begin{itemize}
\item \emph{Claim 1}. Every parallel process of $u[\infi{A}/a]$ is of the form $t[\infi{A}/a]$, where $t$ is  a parallel process of $u$. 
\item \emph{Claim 2}. For every parallel process $t[\infi{A}/a]$ of  $u[\infi{A}/a]$, if the starting symbol of the head redex of $t$ is in the $n$-th, from left to right, elementary process of $u$, then  the starting symbol of the head redex of $t[\infi{A}/a]$ is in the $n$-th elementary process, from left to right, of  $u[\infi{A}/a]$.
\end{itemize}

\noindent\emph{Proof of Claim 1}. By induction on $u$. 
Let $t'$ be a parallel process of $u'=u[\infi{A}/a]$. If $t'=u'$, since $u$ is a parallel process of itself, we are done. Suppose now $u'=u_{1}'\parallel_{b'}u_{2}'$. Then, assuming $u=u_{1}\parallel_{b} u_{2}$, we have
$$u_{1}'\parallel_{b'}u_{2}'= (u_{1}\parallel_{b} u_{2})[\infi{A}/a]=(u_{1}[\infi{A}/a]\parallel_{b} u_{2}[\infi{A}/a])$$
Therefore, $b'=b$, $u_{1}'=u_{1}[\infi{A}/a]$ and $u_{2}'=u_{2}[\infi{A}/a]$. Now, $t'$ is a parallel process of  either $u_{1}'$ or $u_{2}'$; by induction hypothesis, $t'=t[\infi{A}/a]$, where $t$ is a parallel process of $u_{1}$, in the first case, and of $u_{2}$ in the second.  

\noindent\emph{Proof of Claim 2}. By induction on $u$. Let $t'=t[\infi{A}/a]$ be a parallel process of $u'=u[\infi{A}/a]$, where $t$ is a parallel process of $u$. We have two cases.\begin{enumerate}
\item $t'=u[\infi{A}/a]$. If $u= h\, \sigma$ and $h$ is its head redex, then $u$ is the first and unique elementary process of $u$; moreover,  $t'=(h[\infi{A}/a])\,(\sigma[\infi{A}/a])$, thus $h[\infi{A}/a]$ is its head redex and  indeed $t'$ is the first and unique elementary process of $t$. If $u=u_{1}\parallel_{b}u_{2}$ and is a redex, then by Definition \ref{definition-headred}, the starting symbol of $u$ is the occurrence of $b$ such that $b\,w\,\sigma$ is the  $n$-th elementary process of $u$ and for every $m<n$, the $m$-th elementary process of $u$ does not start with $b$. Since $t'=u_{1}[\infi{A}/a]\parallel_{b} u_{2}[\infi{A}/a]$ and $a\neq b$, the starting symbol of $t'$ is the occurrence of $b$ such that $b\,(w[\infi{A}/a])\,(\sigma [\infi{A}/a])$ is the  $n$-th elementary process of $t'$.\\

\item $t'\neq u[\infi{A}/a]$. Since $t'=t[\infi{A}/a]$ and $t$ is a parallel process of $u$ with $t\neq u$, there is a parallel context $\mathcal{C}[\ ]$ such that $u=\mathcal{C}[t]$. By induction hypothesis, assuming that the starting symbol of the head redex of $t$ is in the $n$-th elementary process of $t$, then the starting symbol of the head redex of $t'$ is in the $n$-th elementary process of $t'$. Since $u'=(\mathcal{C}[\infi{A}/a])[t']$, if $m$ is the number of elementary processes on the left of $t$ in $\mathcal{C}[t]$, then the starting symbol of the head redex of $t$ is in the $(m+n)$-th elementary process of $u$ and the starting symbol of the head redex of $t'$ is in the $(m+n)$-th elementary process of $u'$, which is the thesis. 
\end{enumerate}
Let now us return to the main line of the proof. Suppose $$u\parallel_{a} v\redn u'\parallel_{a} v$$
 Then for some parallel context $\mathcal{C}$, we have $u=\mathcal{C}[q]$ and $u'=\mathcal{C}[q']$, where $q'$ is either obtained from $q$ by contracting the head redex $r$ of $q$ or  $q'=v[\lambda y\, t/a]$ and $q=a\,t\, \sigma$; in the first case,  it is $r$ that is  the leftmost among the head redexes of the parallel processes of $u\parallel_{a} v$, whereas in the second case, it is $u\parallel_{a}v$. With this notation,
 we have  
$$\mathcal{C}[q]\parallel_{a} v\redn \mathcal{C}[q']\parallel_{a} v$$
We must show $$u[\infi{A}/a]\succ^{+} u'[\infi{A}/a]$$
There are several cases.
\begin{itemize}
\item $q=(\lambda x\, s) t\, \sigma$ and $r=(\lambda x\, s) t$.   Let $r'=s[t/x]$. We first need to show that $r[\infi{A}/a]$ is the leftmost among the head redexes of  the parallel process of $u[\infi{A}/a]$ as well. Assume that the starting symbol of $r$  is in the $n$-th elementary processes of $u$. By Claim 2,  the starting symbol of $r[\infi{A}/a]$  is in the $n$-th elementary process of $u[\infi{A}/a]$. Suppose by the way of contradiction, that there is a parallel process $p'$ of $u[\infi{A}/a]$  whose head redex has a starting symbol more on the left, that is, in the $m$-th elementary process of $u[\infi{A}/a]$, with $m<n$.  By Claim 1, $p'=p[\infi{A}/a]$, where $p$ is a parallel process of $u$. By Claim 2, the starting symbol of the head redex of $p$ is in the $m$-th elementary process of $u$, which contradicts the assumption on $q$ and $r$.  
Now, letting $s'=s[\infi{A}/a]$, $t'=t[\infi{A}/a]$, $\mathcal{C}'=\mathcal{C}[\infi{A}/a]$, $\sigma'=\sigma[\infi{A}/a]$, we get
$$u[\infi{A}/a]=\mathcal{C}'[(\lambda x\, s') t'\, \sigma']\redn \mathcal{C}'[s'[t'/x]\, \sigma']=\mathcal{C}[s[t/x]\,\sigma][\infi{A}/a]=u'[\infi{A}/a]$$
\item $q=r\,\sigma$ and $r=\pair{s_{0}}{s_{1}}\, \pi_{i}$ or  $r= \inj_{i}(s)[x_{0}.t_{0}, x_{1}.t_{1}]$ or $ r= (m,s)[(\alpha, x).t]$ or  $r=(w_{1}\parallel_{b}w_{2})\,\rho$ and let, respectively, $r'=s_{i}$ or $r'=t_{i}[s/x_{i}]$ or $r'=t[m/\alpha][s/x]$ or $r'= w_{1}\,\rho\parallel_{b} w_{2}\,\rho$. By exactly the same considerations of the previous case, we get
$$u[\infi{A}/a]=\mathcal{C}[\infi{A}/a][r\,\sigma[\infi{A}/a]]\redn \mathcal{C}[\infi{A}/a][r'\,\sigma[\infi{A}/a]]=\mathcal{C}[r'\,\sigma][\infi{A}/a]=u'[\infi{A}/a]$$
\item $q=r=\abort_{k}\,w\, \sigma$ or $q=r=\mathcal{C}_{1}[b\,w\, \rho]\parallel_{b} s$ for some variable $b\neq a$ (the other case is symmetric); let respectively $r'=w$ or $r'=\mathcal{C}_{1}[s[\lambda y\, w / b]]\parallel_{b} s$. By exactly the same considerations of the previous case, we get
$$u[\infi{A}/a]=\mathcal{C}[\infi{A}/a][r[\infi{A}/a]]\redn \mathcal{C}[\infi{A}/a][r'[\infi{A}/a]]=\mathcal{C}[r'][\infi{A}/a]=u'[\infi{A}/a]$$
\item $q'=v[\lambda y\, t/a]$ and $q=a\,t\, \sigma$. 
Let $t'=t[\infi{A}/a]$, $\sigma'=\sigma[\infi{A}/a]$ and $\mathcal{C}'=\mathcal{C}[\infi{A}/a]$. We first need to show that the head redex of $q[\infi{A}/a]=\infi{A}\,t'\,\sigma'$ is the leftmost among the head redexes of  the parallel processes of $u[\infi{A}/a]$. Assume that $a\,t\, \sigma$ is the $n$-th elementary process of $u$, so that $\infi{A}\,t'\,\sigma'$ is the $n$-th elementary process of $u[\infi{A}/a]$ as well. Then, no parallel process of $u$ has an head redex whose starting symbol is in the $m$-th elementary process of $u$, with $m<n$. By Claims 1 and 2, no parallel process $p[\infi{A}/a]$ of $u[\infi{A}/a]$, where $p$ is a parallel process of $u$, has an head redex whose starting symbol is in the $m$-th elementary process of $u$, with $m<n$, otherwise the starting symbol of the head redex of $p$ would be in the $m$-th elementary process of $u$ as well (Claim 2 applies, since $p$ cannot be of the form $a\, w\, \rho$, given that $a$ is the starting symbol of the redex $u\parallel_{a}v$). Finally, we conclude
\[u[\infi{A}/a]=\mathcal{C}'[\infi{A}\,t'\,\sigma']= \mathcal{C}'[(\lambda x\, \abort\,v[\lambda y\, x / a])\, t'\,\sigma']\redn\]
\[ \mathcal{C}'[ \abort\,v[\lambda y\, t' / a]\,\sigma']\redn \mathcal{C}'[v[\lambda y\, t' / a]]=\mathcal{C}[v[\lambda y\, t / a]][\infi{A}/a]=u'[\infi{A}/a]\eqno{\qEd}\]
\end{itemize}\bigskip

\comment{
\newcommand{\n} {{\,|\,}}
\begin{prop}[Simulation with non-determinism]
Suppose $$u[v/x]\redn s$$
and $x$ does not occur in $v,w$. Then there is a term $u'$ such that
$$s=u'[v/x]\ \land\ u[(v\n w)/x]\redn u'[(v\n w)/x]$$
\end{prop}
\begin{proof} There are several cases, according to what kind of redex has been contracted to obtain $s$.
\begin{itemize}
\item The redex contracted is inside an occurrence of $v$ which replaces $x$ in $u$. This means that for some context $\mathcal{C}$, we  have 
$$u[v/x]=\mathcal{C}[x]\, [v/x]=\mathcal{C}'[v]\redn \mathcal{C}'[t]=s$$
with $v\mapsto t$. 
 Let $u'=\mathcal{C}[t]$. Then by construction
$$s=\mathcal{C}'[t]=\mathcal{C}[t]\, [v/x]=u'[v/x]$$
and
$$u[(v\n w)/x]=\mathcal{C}[x]\, [(v\n w)/x]=\mathcal{C}''[v\n w]\redn\mathcal{C}''[v]\redn\mathcal{C}''[t]= u' [(v\n w)/x]$$
To get the thesis, we have just to check that this second reduction is correct. Indeed, $v$ is the leftmost redex in $\mathcal{C}'[v]$, so $x$ cannot occur on the left of the displayed occurrence of $x$ in $\mathcal{C}[x]$; therefore, in $\mathcal{C}''[v\n w]$ the term $v\n w$ is still the leftmost redex.

\item The redex contracted is created by the substitution of  $x$ with $v$ in a subterm $x\,\sigma$ of $u$: we consider only the case $v=\lambda z\, t$, the other cases posing no additional challenges. This means that for some context $\mathcal{C}$, we have 
$$u[v/x]=\mathcal{C}[x\,\sigma]\, [v/x]=\mathcal{C}'[(\lambda z\, t)\,\sigma']\redn \mathcal{C}'[t[\sigma'/z]]=s$$
 Let $u'=\mathcal{C}[t[\sigma/z]]$. Then by construction
$$s=\mathcal{C}'[t[\sigma'/z]]=\mathcal{C}[t[\sigma/z]]\, [v/x]=u'[v/x]$$
and
$$u[(v\n w)/x]=\mathcal{C}[x\,\sigma]\, [(v\n w)/x]=\mathcal{C}''[(v\n w)\,\sigma'']\redn\mathcal{C}''[(\lambda z\, t)\,\sigma'']\redn\mathcal{C}''[t[\sigma''/z]]= u' [(v\n w)/x]$$
To get the thesis, we have just to check that this second reduction is correct. Indeed, $(\lambda z\, t)\sigma'$ is the leftmost redex in $\mathcal{C}'[(\lambda z\, t)\sigma']$, so $(\lambda z\, t)\sigma'$ cannot occur on the left of the displayed occurrence of $v$; therefore, in $\mathcal{C}''[v\n w]$ the term $v\n w$ is still the leftmost redex.
\end{itemize}
\end{proof}
}
\noindent We are now able to tackle the most difficult case of the Adequacy Theorem \ref{Adequacy Theorem} for realizability: proving that realizability is also sound for the Dummett rule. The idea is that Proposition \ref{proposition-simul} allows us to use in a very strong manner an inductive hypothesis that will naturally be granted when proving the Adequacy Theorem.  This hypothesis is knowing that for every $t\red A\rightarrow B$, $u[t/a]\in\snn$; since one can prove that the term $\infi{A}$ realizes $A\rightarrow B$, one can conclude with simple reasoning that the head reduction reduces $u$ in  $u \parallel_{a} v$ only a finite number of times and a symmetric reasoning  holds for $v$. Hadn't we the abort operator and thus the possibility of local simulation, the hypothesis  that for every $t\red A\rightarrow B$, $u[t/a]\in\snn$, would not be enough to conclude a great deal. Details follow.

\begin{prop}[Preservation of Realizability by Parallel Composition]\label{proposition-parallel}\mbox{}
\begin{enumerate}
\item If for every $t\red A\rightarrow B$, $u[t/a]\in\snn$ and for every $t\red B\rightarrow A$, $v[t/a]\in\snn$, then $u \parallel_{a} v \in\snn$.
\item If for every $t\red A\rightarrow B$, $u[t/a]\red C$ and for every $t\red B\rightarrow A$, $v[t/a]\red C$, then $u \parallel_{a} v \red C$.

\end{enumerate}
\end{prop}
\begin{proof}\mbox{}
\begin{enumerate}
\item Define 
$$\infi{A}:= \lambda x\,\abort ( v[\lambda y\, x/a])$$
We start by showing that $\infi{A}\red A\rightarrow B$, which establishes by means of the hypothesis that $u[\infi{A}/a]\in\snn$. 
Let $\rho\in ||A\rightarrow B||$; the case $\rho=\epsilon$ is trivial, so we assume $\rho= t\dott \sigma$, with $t\red A$ and $\sigma\in ||B||$. We must show $\infi{A}\, t\,\sigma\in\snn$. We have 
$$\infi{A}\,t\,\sigma= (\lambda x\,\abort ( v[\lambda y\, x/a])) t\, \sigma\redn \abort ( v[\lambda y\, t/a]) \, \sigma\redn v[\lambda y\, t/a]$$
(assuming the last reduction is possible: if not, the thesis is trivial).
In order to obtain $v[\lambda y\, t/a]\in\snn$, which is what we wanted, it is enough to show that $\lambda y\, t\red B\rightarrow A$.  Let $\rho'\in ||B\rightarrow A||$;  again, the case $\rho'=\epsilon$ is trivial, so we assume $\rho'= t'\dott \sigma' $, with $t'\red B$ and $\sigma'\in ||A||$. We must show $(\lambda y\, t)\, t'\,\sigma'\in\snn$. Indeed, since $t\red A$,
$$(\lambda y\, t)\, t'\,\sigma'\redn t\,\sigma'\in \snn$$

We now prove that $u\parallel_{a} v\in\snn$ by induction on the length of the reduction of $u[\infi{A}/a]$ in head normal form. We have two cases.

\begin{enumerate}
\item  Assume $$u\parallel_{a} v\redn u\parallel_{a} v'$$
so that in particular $u$ is in head normal form. Define 
$$\infi{B}:= \lambda x\,\abort ( u[\lambda y\, x/a])$$
Since we are going again to prove the thesis by induction on the length of the reduction of  $v[\infi{B}/a]$ in head normal form,  we first need to show that $\infi{B}\red B\rightarrow A$, which allows us to conclude that indeed $v[\infi{B}/a]\in\snn$. Let $\rho\in ||B\rightarrow A||$; the case $\rho=\epsilon$ is trivial, so we assume $\rho= t\dott \sigma$, with $t\red B$ and $\sigma\in ||A||$. We have to show $\infi{B}\, t\,\sigma\in\snn$. We have 
$$\infi{B}\,t\,\sigma= (\lambda x\,\abort ( u[\lambda y\, x/a])) t\, \sigma\redn \abort ( u[\lambda y\, t/a]) \, \sigma\redn u[\lambda y\, t/a]$$
Now,  $u$ is in head normal form and thus by Proposition \ref{proposition-hnf},  $$u=u_{0}\parallel_{a_{1}} u_{1}\parallel_{a_{2}}\ldots \parallel_{a_{n}} u_{n}$$
and for each $0 \leq i\leq n$, either $u_{i}=x\, \sigma$, with $x\neq a_{1},\ldots, a_{n}, a$, or $u_{i}=\abort\, w\, \sigma$, with the type of $w$ different from the type of $\abort\, w\, \sigma$, or  $u_{i}=a_{j}$ or $u_{i}=a$ or $u_{i}$ is a value.
Therefore, $u[\lambda y\, t/a]$ is in head normal form, because the substitution does not create head redexes in any parallel process of $u$.

We now prove the main thesis. 
By Proposition \ref{proposition-simul}, $v[\infi{B}/a]\redn^{+} v'[\infi{B}/a]$, so by induction hypothesis we conclude $u\parallel_{a} v'	\in\snn$ and thus $u\parallel_{a} v\in\snn$.

\item Assume
$$u\parallel_{a} v\redn u'\parallel_{a} v$$
By Proposition \ref{proposition-simul}, $u[\infi{A}/a]\redn^{+} u'[\infi{A}/a]$, so by induction hypothesis we conclude $u'\parallel_{a} v\in\snn$ and thus $u\parallel_{a} v\in\snn$.
\end{enumerate}

\item Let $\sigma\in ||C||$. We must show that $(u\parallel_{a} v)\,\sigma\in\snn$. By hypothesis, for every $t\red A\rightarrow B$, $u[t/a]\,\sigma \in\snn$ and for every $t\red B\rightarrow A$, $v[t/a]\,\sigma \in\snn$. By point 1., $u\,\sigma \parallel_{a} v\,\sigma\in\snn$. Since 
$$(u\parallel_{a} v)\,\sigma \redn^{*}  u\,\sigma \parallel_{a} v\,\sigma$$
we are done.\qedhere
\end{enumerate}
\end{proof}

\subsection{The Adequacy Theorem}\label{section-adequacy}
We finally prove that realizability is sound for $\LC$: if we replace all free proof term variables of any proof term with realizers, then we get a realizer. 
\begin{thm}[Adequacy Theorem]\label{Adequacy Theorem}
Suppose that $w: A$ in
the system $\LC$, with $w$ having free variables among
$x_1^{A_1},\ldots,x_n^{A_n}$. 
For all terms $r_1,\ldots,r_k$ of $\Language$, if there are terms $t_1, \ldots, t_n$ such that
$$\text{ for  $i=1,\ldots, n$, }t_i\red  \subst{A}_{i}:= A_i[{r}_1/\alpha_1\cdots
{r}_k/\alpha_k]$$
 then
$$w [{r}_1/\alpha_1\cdots
{r}_k/\alpha_k][t_1/x_1^{\subst{A}_{1}}\cdots
t_n/x_n^{\subst{A}_{n}}]\red A[{r}_1/\alpha_1\cdots
{r}_k/\alpha_k]$$
\end{thm}

\begin{proof}

For any term $v$ and formula $B$, we define 
$$\subst{v}:=v[{r}_1/\alpha_1\cdots
{r}_k/\alpha_k][t_1/x_1^{\subst{A}_{1}}\cdots
t_n/x_n^{\subst{A}_{n}}]$$
and
  $$\subst{B}:=B[{r}_1/\alpha_1\cdots {r}_k/\alpha_k]$$ We proceed by induction on $w$. Consider the last rule $\mathscr{R}$ in the derivation of $w: A$:

\begin{enumerate}[widest=10]

\item
If $\mathscr{R}= x_{i}^{A_{i}}: A_{i}$, for some $i$, then
$w=x_i^{A_{i}}$ and
$A=A_i$. So $\subst{w}=t_i\red
\subst{A}_i=\subst{A}$.


\item If $\mathscr{R}$  is the $\rightarrow E$ rule, then $w=t\, u$, 
$$ t:
B\rightarrow A \qquad u: B$$
 So by Proposition \ref{proposition-redelim},
$\subst{w}=\subst{t}\,\subst{u}\red \subst{A}$, for
$\subst{t}\red \subst{B}\rightarrow \subst{A}$ and
$\subst{u}\red \subst{B}$ by induction hypothesis.

  \item If $\mathscr{R}$ is the $\rightarrow I$ rule, then $w=\lambda x^{B}\,
u$,
$A=B\rightarrow C$ and $ u: C$. So, $\subst{w}=\lambda x^{\subst{B}}\, \subst{u}$, because by renaming of bound variables we can assume $x^{\subst{B}} \neq x_1^{\subst{A}_{1}}, \ldots, x_k^{\subst{A}_{k}}$. For  every  $t\red
\subst{B}$, by induction hypothesis on $u$, $\subst{u}[t/x^{\subst{B}}]\red\subst{C}$.  
Therefore, by Proposition \ref{proposition-somecases}, $\lambda x^{\subst{B}}\, \subst{u}\red \subst{B}\rightarrow \subst{C}=
\subst{A}$.

\item
If $\mathscr{R}$ is a $\vee I$ rule, say left (the other case is symmetric), then $w=\inj_{0}(u)$, $A=B\vee C$ and $ u: B$. So, $\subst{w}=\inj_{0}(\subst{u})$ and by induction hypothesis $\subst{u}\red \subst{B}$. Hence, by Proposition \ref{proposition-somecases} we conclude $\inj_{0}(\subst{u}) \red \subst{B}\lor\subst{C}= \subst{A}$.

\item If $\mathscr{R}$ is  a $\vee E$ rule, then
$$w=  u [x^{B}.w_1, y^{C}.w_2] $$
 and  
 $$ u: B\vee C \qquad  w_1: D\qquad w_2: D$$ 
 with $A=D$.  By induction hypothesis, we have $\subst{u}\red \subst{B}\lor \subst{C}$; moreover,  for every $t\red \subst{B}$, we have $\subst{w}_{1}[t/x^{\subst{B}}]\red \subst{D}$ and for every $t\red \subst{C}$, we have $\subst{w}_{2}[t/y^{\subst{C}}]\red \subst{D}$.  By Proposition \ref{proposition-redelim}, we obtain $\subst{w}=\subst{u}\, [x^\subst{B}.\subst{w}_1, y^\subst{C}.\subst{w}_2]\red \subst{D}$.

 \item The cases $\mathscr{R}=\land E$ and $\mathscr{R}=\land I$ are straightforward. 

\item The cases $\mathscr{R}=\exists I$ and $\mathscr{R}=\exists E$ are similar respectively to $\lor I$ and $\lor E$.

\item If $\mathscr{R}$ is the $\forall E$ rule, then $w=u\, m$, $A=B[m/\alpha]$
and $u: \forall \alpha^{}\, B$. So,
$\subst{w}=\subst{u}\,\subst{m}$.  By inductive hypothesis  $\subst{u}\red
\forall\alpha^{}\, \subst{B}$ and so $\subst{u}\,\subst{m}\red \subst{B}[\subst{m}/\alpha]$ by Proposition \ref{proposition-redelim}.

\item
If $\mathscr{R}$ is the  $\forall I$ rule, then $w=\lambda \alpha\, u$, $A=\forall \alpha\, B$ and $u: B$ (with $\alpha$ not occurring free in the types $A_{1}, \ldots, A_{n}$ of the free variables of $u$). So, $\subst{w}=\lambda \alpha\, \subst{u}$, since we may assume $\alpha\neq \alpha_1, \ldots, \alpha_k$. Let $m$ be a term of $\Language$; by Proposition \ref{proposition-somecases}, it is enough to prove that $\subst{u}[m/\alpha]\red \subst{B}[{m}/\alpha]$, which amounts to showing that the induction hypothesis can be applied to $u$. For this purpose, we observe that, since $\alpha\neq \alpha_1, \ldots, \alpha_k$, for $i=1, \ldots, n$ we have
$$t_i\red \subst{A}_i=\subst{A}_i[m/\alpha]$$

\comment{
  \item If it is the  $\exists E$ rule, then \[w=(\lambda
\alpha^{\tau}\lambda x^{|B|} t)(u\pi_0)( u\pi_1)\] $\Gamma, x:B \vdash t: A$ and
$\Gamma \vdash u: \exists \alpha^{\tau}. B$. Assume ${v} = \pi_0
\subst{u}[s]$. Then
\[\subst{t}[{v}/\alpha^{\tau},\pi_1\subst{u}/
x^{|{B}|}]\red
\subst{A}[{v}/\alpha^{\tau}]=\subst{A}\] by
inductive hypothesis, whose application being justified by the
fact, also by induction, that $\subst{u}\red \exists
\alpha^{}.
\subst{B}$ and hence $\pi_1\subst{u}\red
\subst{B}[{v}/\alpha^{\tau}]$. We thus obtain by
$ \subst{w} [s] =\subst{t}[\pi_0 \subst{u}/\alpha^{\tau}\ \pi_1\subst{u}/x^{|B|}] [s] $ and saturation (prop. \ref{proposition-saturation}) that
\[\subst{w} \red \subst{A}\]

  \item If \mathscr{R} is the $\exists I$ rule, then $w=\langle  t,u\rangle
$, $A=\exists
\alpha^{\tau}
B$, $\Gamma \vdash u: B[t/\alpha^{\tau}]$. So, $\subst{w}=\langle
\subst{t},\subst{u}\rangle $; and, indeed, $\pi_1
\subst{w}=\subst{u}\red \subst{B}[\pi_0\subst{w}/\alpha^{\tau}]=\subst{B}[\subst{t}/\alpha^{\tau}]$ since by induction hypothesis
$\subst{u}\red \subst{B}[\subst{t}/\alpha^{\tau}]$. By saturation we conclude the thesis. \\
}

\item If $\mathscr{R}$ is the $\mathsf{D}$ rule, then $w= u\parallel_{a} v$, $A=D$ and 
\begin{prooftree}
\AxiomC{$[a^{\scriptscriptstyle B\rightarrow C}: B\rightarrow C]$}
\noLine
\UnaryInfC{$\vdots$}
\noLine
\UnaryInfC{$u: D$}
\AxiomC{$[a^{\scriptscriptstyle C\rightarrow B}: C\rightarrow B]$}
\noLine
\UnaryInfC{$\vdots$}
\noLine
\UnaryInfC{$v: D$}
\BinaryInfC{$u \parallel_{a} v: D$}
\end{prooftree}
By induction hypothesis, for every $t\red \subst{B}\rightarrow \subst{C}$, we have $\subst{u}[t/a]\red \subst{D}$ and for every $t\red \subst{C}\rightarrow \subst{B}$,  $\subst{v}[t/a]\red \subst{D}$. By  Proposition \ref{proposition-parallel}, we conclude $\subst{w}=\subst{u}\parallel_{a}\subst{v}\red \subst{D}$. 

\item If $\mathscr{R}$ is the ex falso quodlibet rule, then $w=\efq{P}{u}$, $A=P$ and $u: \bot$. 
Now, $||P||=\{\epsilon\}$ and $\subst{w}\,\epsilon=\subst{w}=\efq{P}{\subst{u}} \in\snn$. We conclude $\subst{w}\red A$. \qedhere
\end{enumerate}
\end{proof}

\subsection{Normalization 
for $\LC$}

As corollary of the Adequacy Theorem \ref{Adequacy Theorem}, one obtains normalization for $\LC$.

\begin{cor}[Normalization for $\LC$]\label{theorem-snNEM} Suppose that  $ t: A$ is a proof term of $\LC$. Then $t\in\snn$.
\end{cor}
\begin{proof}
Assume $x_1: {A_1},\ldots,x_n:{A_n}$ are the free variables of $t$. We observe that $x_{i}\red A_{i}$, for $i=1,\ldots, n$ because, given any $\sigma\in ||A_{i}||$, $x\,\sigma\in \snn$. Therefore, 
from Theorem \ref{Adequacy Theorem}, we derive that $t\red A$ and since $\epsilon \in ||A||$, we conclude $t= t\,\epsilon \in\snn$. 
\end{proof}

\comment{
As corollary, one obtains strong normalization for $\ALC$.
\begin{cor}[Normalization for $\ALC$] Suppose $t: A$ in $\ALC$. Then $t\in\snn$.
\end{cor}
\begin{proof}
Assume $x_1: {A_1},\ldots,x_n:{A_n}$ are the free variables of $t$. We observe that $x_{i}\red A_{i}$, for $i=1,\ldots, n$ because, given any $\sigma\in ||A_{i}||$, $x\,\sigma\in \snn$. Therefore, 
from Theorem \ref{Adequacy Theorem}, we derive that $t\red A$ and since $\epsilon \in ||A||$, we conclude $t= t\,\epsilon \in\snn$. 
\end{proof}

The strong normalization of $\ALC$ is readily turned into a strong normalization result for $\LC$, since every term typable in $\LC$ is obviously typable in $\ALC$, for $\ALC$ is an extension of $\LC$. 

\begin{cor}[Normalization for $\LC$]\label{theorem-snNEM} Suppose $ t: A$ in $\LC$. Then $t\in\sn$.
\end{cor}
}

\section{
Normal Form Property and Herbrand's Disjunction Extraction}\label{sec:nfhd}

In this section, we finally show that our Curry-Howard correspondence for $\LC$ is meaningful from the computational perspective. 
We already know that every  execution of every program we extract always terminate; now we prove that in the case of \emph{any} existentially quantified formula $\exists \alpha\, A$, every closed program of that type produces a complete finite sequence $m_{1}, m_{2}, \ldots, m_{k}$ of possible witnesses for $\exists \alpha\, A$. This means that whatever first-order model we consider, there will be an $i$  such that $A[m_{i}/\alpha]$ is true in it. In other terms, we have provided a proof that $\LC$ is Herbrand constructive and a Curry-Howard computational interpretation of this very strong Herbrand-like theorem. 

Such  statements in first-order logic are typically drawn as consequences of the Subformula Property, which is in turn a corollary of full cut-elimination  when sequent calculus is available.  
But as in \cite{AschieriZH}, a more primitive argument suffices here. This is indeed providential, since not only without permutation rules for $\lor$ and $\exists$ we can have no Subformula Property, but surprisingly even those reductions would not suffice. The topic of what reductions are needed is very non-trivial and left as subject of future research.
However, in a sense, Herbrand constructiveness is already a \emph{weak} Subformula Property and holds for the most interesting case of the existential quantifier, when there is actually some information to gain. For lambda calculus, instead,  to enjoy the Subformula Property is a mere curiosity without much computational sense. In fact, if we think that in intuitionistic Logic or fragments of classical Arithmetic \cite{ABB} general permutation rules are not needed to compute witnesses, it should not entirely come as a surprise that this is still the case in our framework.


If we omit parentheses, we know that every proof term in head normal form can be written as $v_{0}\parallel_{a_{1}} v_{1}\ldots \parallel_{a_{n}} v_n$, where each $v_{i}$ is not of the form ${u}\parallel_{a}{v}$; if for every $i$, $v_i$ is of the form $(m_i,u_i)$, then we call the whole term an \emph{Herbrand normal form}, because it is essentially a list of the witnesses  appearing in an Herbrand disjunction. Formally:
\begin{defi}[Herbrand Normal Forms] \label{definition-hnf}
We define by induction a set of proof terms, called \textbf{Herbrand normal forms}, as follows:
\begin{itemize}
\item Every proof-term $(m,u)$ is an Herbrand normal form;
\item if $u$ and $v$ are Herbrand normal forms, $u\parallel_{a} v$ is an Herbrand normal form.
\end{itemize}
\end{defi}

An Herbrand normal form represents, in a straightforward way, a proof of an Herbrand disjunction. 

\begin{prop}[Herbrand Normal Forms and Herbrand Disjunctions]\label{proposition-herbrandforms}\mbox{}\\
Suppose that  $\Gamma\vdash u: \exists \alpha\, A$ in $\LC$  and $u$ is an Herbrand normal form 
$$(m_{0}, v_{0})\parallel_{a_{1}} (m_{1}, v_{1})\parallel_{a_{2}}\ldots \parallel_{a_{k}} (m_{k}, v_{k})$$
Then for some $u^{+}$ 
$$ \Gamma\vdash u^{+}: A[m_{1}/\alpha]\lor \dots \lor A[m_{k}/\alpha]$$
\end{prop}
\begin{proof}
We proceed by induction on $k$. \\
If $k=0$, then $u=(m_{0}, v_{0})$ and thus $\Gamma \vdash v_{0}: A[m_{0}/\alpha]$, which is the thesis. \\
If $k>0$, then $u=w_{1}\parallel_{a_{i}} w_{2}$,  for some $1\leq i\leq n$ and 
$$w_{1}=(m_{0}, v_{0})\parallel_{a_{1}} (m_{1}, v_{1})\parallel_{a_{2}}\ldots \parallel_{a_{i-1}} (m_{i-1}, v_{i-1})$$
$$w_{2}=(m_{i}, v_{i})\parallel_{a_{i+1}} (m_{i+1}, v_{i+1})\parallel_{a_{i+2}}\ldots \parallel_{a_{k}} (m_{k}, v_{k})$$
\begin{prooftree}
\AxiomC{$[a_{i}^{\scriptscriptstyle B\rightarrow C}: B\rightarrow C]$}
\noLine
\UnaryInfC{$\vdots$}
\noLine
\UnaryInfC{$w_{1}: \exists \alpha\, A$}
\AxiomC{$[a_{i}^{\scriptscriptstyle C\rightarrow B}: C\rightarrow B]$}
\noLine
\UnaryInfC{$\vdots$}
\noLine
\UnaryInfC{$w_{2}: \exists\alpha\, A$}
\BinaryInfC{$w_{1} \parallel_{a} w_{2}: \exists\alpha\, A$}
\end{prooftree}
By induction hypothesis, 
$$\Gamma, a_{i}: B\rightarrow C\vdash w_{1}^{+}: A[m_{1}/\alpha]\lor \dots \lor A[m_{i-1}/\alpha]$$
$$\Gamma, a_{i}: C\rightarrow B\vdash w_{2}^{+}: A[m_{i}/\alpha]\lor \dots \lor A[m_{k}/\alpha]$$
Hence, by repeated application of the $\lor I$ inference, we get, for some $s_{1}, s_{2}$,
$$\Gamma, a_{i}: B\rightarrow C\vdash s_{1}: A[m_{1}/\alpha]\lor \dots \lor A[m_{k}/\alpha]$$
$$\Gamma, a_{i}: C\rightarrow B\vdash s_{2}: A[m_{1}/\alpha]\lor \dots \lor A[m_{k}/\alpha]$$
and thus 
$$\Gamma \vdash s_{1}\parallel_{a_{i}} s_{2}: A[m_{1}/\alpha]\lor \dots \lor A[m_{k}/\alpha]$$
By setting $u^{+}:=s_{1}\parallel_{a_{i}} s_{2}$, we obtain the thesis. 
\end{proof}

Our last task is to prove that every closed realizer of any existentially quantified statement $\exists \alpha\, A$ include an exhaustive sequence $m_{1}, m_{2}, \ldots, m_{k}$ of possible witnesses.


\begin{thm}[Herbrand Disjunction and Realizability]\label{theorem-extraction0}
Let $\exists\alpha\,A$ be any formula. Suppose $t\red \exists \alpha\, A$, $t$ contains neither free proof term variables nor $\abort$, and $t\succ^{*} u\in\hnf$. Then  $u$ is an Herbrand normal form   
$$u=(m_{0}, v_{0})\parallel_{a_{1}} (m_{1}, v_{1})\parallel_{a_{2}}\ldots \parallel_{a_{k}} (m_{k}, v_{k})$$
and
$$\LC \vdash A[m_{1}/\alpha]\lor \dots \lor A[m_{k}/\alpha]$$

\end{thm}

\begin{proof}\mbox{} By Proposition \ref{proposition-hnf}  
$$u=u_{0}\parallel_{a_{1}} u_{1}\parallel_{a_{2}}\ldots \parallel_{a_{k}} u_{k}$$
where for each $0 \leq i\leq k$, either $u_{i}=x\, \sigma$, with $x\neq a_{1},\ldots, a_{n}$ or  $u_{i}=a_{j}$, with $1 \leq j\leq k$, or $u_{i}$ is a value. Since $u_{i}$ does not contain free proof term variables other than $a_{1},\ldots, a_{k}$, it cannot be of the form $u_{i}=x\, \sigma$. 
Moreover, $u_{i}: \exists \alpha\, A$, hence $u_{i}$ cannot be equal to some $a_{j}$, because $a_{j}$ must have type $B\rightarrow C$. Therefore $u_{i}$ is a value, according to Definition \ref{definition-value}, and the only possible shape compatible with its type $\exists \alpha\, A$ is $(m_{i}, u_{i})$.
We have thus shown that $u$ is an Herbrand normal form
$$(m_{0}, v_{0})\parallel_{a_{1}} (m_{1}, v_{1})\parallel_{a_{2}}\ldots \parallel_{a_{k}} (m_{k}, v_{k})$$
By Proposition \ref{proposition-herbrandforms}, for some $u^{+}$
$$\LC \vdash u^{+}: A[m_{1}/\alpha]\lor \dots \lor A[m_{k}/\alpha]$$
which is the thesis.
\end{proof}

As corollary, we obtain that Dummett's logic $\LC$ is Herbrand constructive. 
\begin{cor}[Herbrand Disjunction Extraction]\label{theorem-extraction2}
Let $\exists\alpha\,A$ be any formula. Suppose 
$$\LC \vdash t:  \exists \alpha\, A$$ Then there is a proof term $u$ such that $t\succ^{*} u\in\hnf$, $\LC \vdash u: \exists \alpha\, A$ and $u$ is an Herbrand normal form   
$$u=(m_{0}, v_{0})\parallel_{a_{1}} (m_{1}, v_{1})\parallel_{a_{2}}\ldots \parallel_{a_{k}} (m_{k}, v_{k})$$
Moreover, 
$$\LC \vdash A[m_{1}/\alpha]\lor \dots \lor A[m_{k}/\alpha]$$

\end{cor}

\begin{proof}\mbox{}By the Subject Reduction Theorem \ref{subjectred}, $\LC\vdash u: \exists \alpha\, A$. By the Adequacy Theorem \ref{Adequacy Theorem}, $t\red \exists \alpha\, A$ and the thesis follows from Theorem \ref{theorem-extraction0}.
\end{proof}
We suggest to interpret an Herbrand normal form $$(m_{0}, v_{0})\parallel_{a_{1}} (m_{1}, v_{1})\parallel_{a_{2}}\ldots \parallel_{a_{k}} (m_{k}, v_{k})$$ in the following way. Each $(m_{i}, u_{i})$ represents the result of an intuitionistic computation of a witness in a possible universe. These witnesses have been obtained by communication coming from other intuitionistic computations in other parallel universes. It is that process of interaction and dialogue between different possible computations that generates the Herbrand normal forms.

\subsection{Parallel Reductions}
Head reduction, of course, is sequential computation. Yet, the operator $\parallel_{a}$ has such a strong parallel flavour that  parallel reduction strategies inevitably arise as consequence of Normalization for head reduction. To see this, let us consider a proof term $u\parallel_{a} v$ of $\LC$. By the Normalization Theorem \ref{theorem-snNEM}, the head reduction of $u\parallel_{a} v$ reduces subterms inside the left part of the term until it is possible, afterwards it continues to reduce the right part and finally it stops. If we consider only the first half of the reduction, we get  
\begin{equation}\label{1} u\parallel_{a}v\succ^{*} u' \parallel_{a} v\end{equation}
for some $u'$ in head normal form and \emph{not} of the shape $\mathcal{C}[a\, t\,\sigma]$, otherwise a further reduction inside $u'$ would be possible. Thanks to the \emph{perfect logical symmetry} of the term $u\parallel_{a} v$, also $v\parallel_{a} u$ is a term of the same type. Again, we can reduce 
\begin{equation}\label{2}v\parallel_{a}u\succ^{*} v' \parallel_{a} u\end{equation}
for some $v'$ in head normal form and \emph{not} of the shape $\mathcal{C}[a\, t\,\sigma]$.
The point is that the head reductions (\ref{1}) and (2) \emph{can be made in parallel} and what we get,
\[u'\parallel_{a} v'\]
not only is a term of the same type of $u\parallel_{a} v$ and with no more free variables, it also
is a head normal form!

\section{Second-Order Intuitionistic Logic with Dummett's Axiom}\label{section-secondorder}

At the time of this writing, there is no known cut-free sequent calculus for second-order intuitionistic logic with Dummett's Axiom, which we call $\LCS$. Even if there were one, the situation would be similar to what happens in the hypersequent calculus for second-order G\"odel-Dummett logic \cite{LAvron}: there is no known cut-elimination procedure,  only a semantical proof that  valid statements can be proved without cuts. Why? This state of thing reminds  the status of Takeuti's conjecture \cite{Takeuti}, a problem which resisted the effort of the best researchers for many years in the 1950-60's, and was solved constructively in 1971 by Girard (see \cite{Girard}). It asked whether the now standard second-order sequent calculus was cut-free. A cut-elimination procedure for intuitionistic second-order sequent calculus was finally obtained only through translation to natural deduction, where the powerful Tait-Girard reducibility settles the matter. This shortcoming of sequent calculus is even worse in the case of hypersequent calculus, which is more complicated and no cut-elimination procedure is known at second-order.

In this section, we consider second-order natural deduction for $\LCS$ and prove the Normalization of head reduction. Unlike in hypersequent calculus, where second-order cut-elimination requires climbing a steep and cold combinatorial mountain, extending classical realizability to the second-oder case is a like a quiet stroll in a peaceful and sunny countryside road. Indeed, classical realizability was introduced directly in the second-order case by Parigot \cite{Parigot} and Krivine \cite{Krivine0}, without even bothering with the first-order case. We follow once again Krivine's successive formulation \cite{Krivine}. 

The language $\Language^{2}$ of $\LCS$ extends $\Language$ in the standard way, adding second order predicate variables, representing sets of individuals. 

\begin{defi}[Language of $\LCS$]\label{definition-languagear}
The language $\Language^{2}$ of $\LCS$ is defined as follows.
\begin{enumerate}

\item
The \textbf{terms} of $\Language^{2}$ are inductively defined as either variables $\alpha, \beta,\ldots$ or constants  $\con{c}$ or expressions of the form $\con{f}(m_{1}, \ldots, m_{n})$, with $\con{f}$ a function constant of arity $n$ and $m_{1}, \ldots, m_{n}\in\Language^{2}$.  

\item
There is a set of \textbf{predicate constant symbols} and of \textbf{predicate variables}. The \textbf{atomic formulas} of $\Language^{2}$ are all the expressions of the form $\mathcal{P}(m_{1}, \ldots, m_{n})$  and $X(m)$ such that  $\mathcal{P}$ is a predicate symbol of arity $n$, $X$ is a predicate variable and $m, m_{1}, \ldots, m_{n}$ are terms of $\Language^{2}$. We assume to have a $0$-ary predicate symbol $\bot$ which represents falsity.

\item
The \textbf{formulas} of $\Language^{2}$ are built from atomic formulas of $\Language^{2}$ by the logical constants  $\lor,\land,\rightarrow, \forall,\exists$, with quantifiers ranging over first-order  variables $\alpha, \beta, \ldots$ and second-order variables $X, Y, \ldots$: if $A, B$ are formulas, then $A\land B$, $A\lor B$, $A\rightarrow B$, $\forall \alpha\, A$, $\exists \alpha\, B$, $\forall X\, A$ are formulas. The  logical negation $\lnot A$ can be introduced, as usual, as an abbreviation of the formula $A\rightarrow\bot$ and the second-order existential quantification is defined as $\exists X\, A:= \forall Y.\,(\forall\,X.\, A\rightarrow Y(\con{c}))\rightarrow Y(\con{c})$.

\item As usual, if $A$ and $B$ are formulas of $\Language^{2}$ and $X$ is a predicate variable, we denote with $A[\lambda \alpha B/X]$ the formula obtained from $A$ by replacing all its atomic subformulas of the form $X(m)$ with $B[m/\alpha]$ (without capturing free variables of $B$).
\end{enumerate}
\end{defi}

\noindent The natural deduction for $\LCS$ and $\ALCS$ extends respectively the natural deduction for $\LC$ and the one for $\ALC$ with the following inference and reduction rules (see Girard \cite{Girard}):\\

\begin{description}
\item[Second-Order Universal Quantification] 
\AxiomC{$t: A$}
\UnaryInfC{$\Lambda X\, t: \forall X\, A$}
\DisplayProof
\ \ \ \ \ \ 
\AxiomC{$t: \forall X\, A$}
\UnaryInfC{$t (\lambda \alpha B): A [\lambda \alpha B/X]$}
\DisplayProof\\\\
where in the left rule $X$ does not occur free in the types of the free variables of $A$.
\end{description}

\begin{description}
\item[Reduction Rule for Universal Quantification] 
$(\Lambda X\, u) (\lambda\alpha B) \mapsto u[\lambda \alpha B/X]$\\
\end{description}
The Definition \ref{definition-stack} of stack is of extended to $\ALCS$ allowing expressions $(\lambda \alpha B)$, with $B$ formula, to appear in the stack,  and the Definition \ref{definition-head} of head redex is extended to the terms of $\ALCS$ by saying that $(\Lambda X\, u) (\lambda\alpha B)$ is the \textbf{head redex} of $(\Lambda X\, u) (\lambda\alpha B)\,\sigma$ for every stack $\sigma$. The reduction relation $\succ$ for the terms of $\ALCS$ is then defined as in Definition \ref{definition-headred}. In the following,  we define $\snn$ to be the set of  normalizing proof terms of $\ALCS$. 

In order to define second-order realizability we need the concept of \emph{realizability opponent}, which is nothing but a function mapping terms of $\Language^{2}$ to arbitrary sets of stacks adapted to some fixed type. The idea is that an arbitrary realizability opponent represents the sets of tests that an arbitrary definition of realizability requires to pass in order to declare a term to be a realizer. 

\begin{defi}[Realizability Opponent]\label{definiton-redo}\mbox{}
\begin{enumerate}
\item
A stack $\sigma$  of $\ALCS$ is said to be \textbf{adapted to a type} $C$, if for all terms $t$ of type $C$, $t\,\sigma$ is still a term of $\ALCS$. 
\item 
A \textbf{realizability opponent} of type $\lambda \alpha C$ is any function that maps each term $m$ of $\Language^{2}$ to a set of stacks adapted to $C[m/\alpha]$. We assume that for each realizability opponent $\mathcal{X}$ of type $C$ there is in $\Language^{2}$  an \textbf{opponent predicate constant} $\cand{\mathcal{X}}$ of type $\lambda \alpha C$ associated to it. 
\end{enumerate}
\end{defi}

\noindent Realizability for $\ALCS$ extends realizability for $\ALC$ to second-order quantification. The idea is the usual: we would like to define $t\red \forall X\, A$ as: for all formulas $B$, $t\,(\lambda \alpha B) \red A[\lambda \alpha B/X]$, but we cannot. So we define $t\red \forall X\, A$ as $t\red A$ for all possible definitions of realizability which $X$ can be assigned to, that is, for all reducibility opponents that replace $X$.

\begin{defi}[Classical Realizability for $\ALCS$]
\label{definition-realizability}
Assume $t$ is a term of $\ALCS$ and $C$ is a formula of $\Language^{2}$. We define by mutual induction the relation $t\red C$ (``$t$ is reducible of type $C$'')  and a set $|| C ||$ of stacks  of $\ALCS$ according to the form of $C$:
\begin{itemize}

\item $t\red C$ if and only if $t: C$ and for all $\sigma\in ||C||$, $t\,\sigma \in \snn$
\item
$||\emp{}|| = \{\epsilon\}$ if $\emp{}$ is atomic

\item $||\cand{{B}}(m) ||=\mathcal{B}(m)$ for each realizability opponent $\mathcal{B}$

\item
$||A\rightarrow B||=\{u\dott\sigma\ |\  u\red A \land \sigma\in ||B||\}\cup \{\epsilon\}$

\item
$||A\land B||=\{\pi_0\dott \sigma\ |\  \sigma\in ||A||\}\cup \{\pi_1\dott \sigma\ |\  \sigma \in ||B||\}\cup \{\epsilon\}$ 

\item
$||A\lor B||=\{ [x.u, y.v]\dott\sigma \ |\ \forall t.\ (t \red A\implies u[t/x]\,\sigma \in \snn) \land  (t\red B\implies v[t/y]\, \sigma \in \snn)\}\cup \{\epsilon\}$ 
\item
$||\forall \alpha\, A||=\{m\dott\sigma\ |\ m\in\Language \land \sigma \in ||A[{m}/\alpha]||\}\cup \{\epsilon\}$
\item

$||\exists \alpha\,A||=\{ [(\alpha, x).v]\dott\sigma \ |\ \forall t.\ t \red A[m/\alpha]\implies v[m/\alpha][t/x]\, \sigma \in \snn \}\cup \{\epsilon\}$
\item $||\forall X\, A||=\{ (\lambda \alpha B)\dott \sigma\ |\  \mbox{$\sigma\in ||A[\cand{B}/X]|| $ for some  realizability opponent $\mathcal{B}$ of type $\lambda \alpha B$}\}\cup \{\epsilon\}$
\end{itemize}
\end{defi}

\noindent The next proposition says that in the definition of $||A[\lambda \alpha B/X]||$, we can replace $\lambda \alpha B$ with the realizability opponent corresponding to it, transforming in this way an intensionally defined set into an extensionally defined object.

\begin{prop}[Comprehension]\label{proposition-comprehension}
Let $B$ a formula of $\Language^{2}$.  Suppose $\mathcal{B}$ is a realizability opponent such that
$$\mathcal{B}= m \mapsto ||B[m/\alpha]||$$
Then for every formula $A$ of $\Language^{2}$
$$||A[\cand{B}/X]||=||A[\lambda \alpha B/X]||$$
\end{prop}
\begin{proof}
Standard, by induction on $A$ (see Krivine \cite{Krivine}). 
\begin{enumerate}
\item $A=\mathcal{P}(m_{1}, \ldots, m_{n})$, where $\mathcal{P}$ is a predicate constant symbol. Then, $A[\cand{B}/X]= \mathcal{P}(m_{1}, \ldots, m_{n})=A[\lambda \alpha B/X]$ and the thesis is trivial.
\item $A=Y(m)$, where $Y$ is a predicate variable. Then, if $Y\neq X$, the thesis is trivial, since we  have 
$$A[\cand{{B}}/X]=Y(m)= A[\lambda \alpha B/X]$$
 So let us suppose $Y=X$. Then
$$
||A[\cand{{B}}/X]||=||\cand{{B}}(m)||=\mathcal{B}(m)=||B[m/\alpha]||=||A[\lambda \alpha B/X]||$$

\item The other cases are straightforward.\qedhere
\end{enumerate}
\end{proof}

\noindent We extended Proposition \ref{proposition-redelim} by showing that realizability is also sound with respect to second-order quantification elimination. 
\begin{prop}[Properties of Realizability: $\forall$-Eliminations]\label{proposition-redelim2}
\item If $t\red \forall X\, A$, then for every formula $B$ of $\Language^{2}$, $t (\lambda \alpha B)\red A[\lambda\alpha B/X]$.

\end{prop}

\begin{proof}\mbox{}
Assume $t\red \forall X\, A$. Let $\sigma \in ||A[\lambda\alpha B/X]||$; we must show $t (\lambda \alpha B)\,\sigma\in\snn$. Let us consider a realizability opponent $\mathcal{B}$ such that
$$\mathcal{B}=m\mapsto ||B[m/\alpha]||$$
By Proposition \ref{proposition-comprehension}, $\sigma\in  ||A[\cand{B}/X]||$. 
 By Definition \ref{definition-realizability}, $(\lambda \alpha B)\dott \sigma \in ||\forall X A||$, and since $t\red \forall X A$,  we conclude $t (\lambda \alpha B)\,\sigma\in\snn$.
\end{proof}
We extended Proposition \ref{proposition-somecases} by showing that realizability is also sound with respect to second-order quantification introduction. 

\begin{prop}[Properties of Realizability: $\forall$-Introductions]\label{proposition-somecases2}

 If for every realizability opponent $\mathcal{B}$ of type $\lambda \alpha B$, $u[\lambda\alpha B/X]\red A[\cand{B}/X]$, then $\Lambda X\, u\red \forall X\, A$.

\end{prop}
\begin{proof}\
Suppose that for every formula $B$  of $\Language^{2}$, $u[\lambda\alpha B/X]\red A[\cand{B}/X]$. Let $\sigma \in ||\forall X A||$. We have to show $(\Lambda X\, u)\,\sigma\in\snn$.  If $\sigma=\epsilon$, indeed $\Lambda X\, u\in\snn$.  Suppose then $\sigma= (\lambda \alpha B)\dott \rho$, with $\rho\in ||A[\cand{B}/X]||$ and $\mathcal{B}$ realizability opponent  of type $\lambda \alpha B$.  Since by hypothesis $u[\lambda\alpha B/X]\red A[\cand{B}/X]$, we have $u[\lambda\alpha B/X]\,\rho\in \snn$; moreover, $$(\Lambda X\, u) (\lambda \alpha B)\, \rho \redn u[\lambda \alpha B/X]\, \rho$$ Therefore, $(\Lambda x\, u) (\lambda\alpha B)\, \rho\in\snn$.
\end{proof}
The Adequacy Theorem is readily extended to second-order realizability. 
\begin{thm}[Adequacy Theorem]\label{Adequacy Theorem2}
Suppose that $w: A$ in
the system $\LCS$, with $w$ having free variables among
$x_1^{A_1},\ldots,x_n^{A_n}$. 
Let $r_1,\ldots,r_k$ and $\cand{B}_{1},\ldots, \cand{B}_{m}$ be respectively terms of $\Language^{2}$ and realizability opponents of type $\lambda \beta_{1} B_{1},\ldots,\lambda \beta_{m} B_{m}$. For every formula $C$, set $\subst{C}=C[{r}_1/\alpha_1\cdots
{r}_k/\alpha_k\  \cand{B}_{1}/X_{1}\cdots  \cand{B}_{m}/X_{m}]$. If there are terms $t_1, \ldots, t_n$ such that
$$\text{ for  $i=1,\ldots, n$, }t_i\red  \subst{A}_{i}$$
 then
$$w [{r}_1/\alpha_1\cdots
{r}_k/\alpha_k\ \lambda \beta_{1} B_{1}/X_{1}\cdots \lambda \beta_{m} B_{m}/X_{m}][t_1/x_1^{\subst{A}_{1}}\cdots
t_n/x_n^{\subst{A}_{n}}]\red \subst{A}$$
\end{thm}
\proof
For any term $v$, we define 
$$\subst{v}:=v[{r}_1/\alpha_1\cdots
{r}_k/\alpha_k\ \lambda \beta_{1} B_{1}/X_{1}\cdots \lambda \beta_{m} B_{m}/X_{m}][t_1/x_1^{\subst{A}_{1}}\cdots
t_n/x_n^{\subst{A}_{n}}]$$
 We proceed by induction on $w$. Consider the last rule $\mathscr{R}$ in the derivation of $w: A$: we just have to deal with the second-order cases, the other ones have been settled in the proof of Theorem \ref{Adequacy Theorem}.
 \begin{enumerate}
 \item If $\mathscr{R}$ is the second-order $\forall E$ rule, then $w=u\, (\lambda\alpha B)$, $A=C[\lambda \alpha B /X]$
and $u: \forall X\, C$. So,
$\subst{w}=\subst{u}\,(\lambda \alpha \subst{C})$.  By inductive hypothesis  $\subst{u}\red
\forall X^{}\, \subst{C}$ and so $\subst{u}\,(\lambda \alpha\subst{B})\red \subst{C}[\lambda \alpha \subst{B}/X]$ by Proposition \ref{proposition-redelim2}. 

\item
If $\mathscr{R}$ is the second-order $\forall I$ rule, then $w=\Lambda X\, u$, $A=\forall X\, B$ and $u: B$ (with $X$ not occurring free in the types $A_{1}, \ldots, A_{n}$ of the free variables of $u$). So, $\subst{w}=\Lambda X\, \subst{u}$, since we may assume $X\neq X_1, \ldots, X_m$. By Proposition \ref{proposition-somecases2}, it is enough to prove that $\subst{u}[\lambda \alpha B/X]\red \subst{B}[\cand{B}/X]$ for every realizability opponent $\mathcal{B}$ of type $\lambda \alpha B$, which amounts to showing that the induction hypothesis can be applied to $u$. For this purpose, we observe that, since $X\neq X_1, \ldots, X_m$, for $i=1, \ldots, n$ we have
\[t_i\red \subst{A}_i=\subst{A}_i[\cand{B}/X]\eqno{\qEd}\]
 \end{enumerate}\medskip

\noindent As consequence of the Adequacy Theorem \ref{Adequacy Theorem2}, we obtain that every typed term of $\LCS$ is normalizable by head reduction. 
\begin{cor}[Normalization for $\LCS$]\label{theorem-sn2order} Suppose $ t: A$ in $\LCS$. Then $t\in\sn$.
\end{cor}
We can finally prove that second-order Dummett's logic $\LCS$ is Herbrand constructive. 
 
\begin{thm}[Second-Order Herbrand Disjunction Extraction]\label{theorem-extraction}
Let $\exists\alpha\,A$ be any formula. Suppose 
$$\LCS \vdash t:  \exists \alpha\, A$$ Then there is a proof term $u$ such that $t\succ^{*} u\in\hnf$, $\LCS \vdash u: \exists \alpha\, A$ and $u$ is an Herbrand normal form   
$$u=(m_{0}, v_{0})\parallel_{a_{1}} (m_{1}, v_{1})\parallel_{a_{2}}\ldots \parallel_{a_{k}} (m_{k}, v_{k})$$
Moreover, 
$$\LCS \vdash A[m_{1}/\alpha]\lor \dots \lor A[m_{k}/\alpha]$$
\end{thm}
\begin{proof}
As the proof of Theorem \ref{theorem-extraction0}.
\end{proof}

\section*{Acknowledgments} I would like to thank Agata Ciabattoni: this work arose, and greatly benefited, from conversations with her. I would also like to thank Francesco Genco for interesting exchanges about the topic.

\appendix


\end{document}